\documentclass[a4paper,12pt]{article}
\usepackage[margin=2.5cm]{geometry}
\usepackage{amsmath,amssymb,amsfonts,amsthm}
\usepackage{graphicx}
\usepackage{setspace}
\usepackage{natbib}
\usepackage[linesnumbered,ruled]{algorithm2e}

\newcommand{\I}{\ensuremath{\mathbb{I}}}
\newtheorem{theorem}{Theorem}
\newtheorem{lemma}[theorem]{Lemma}

\newtheorem{example}{Example}

\newtheorem{assumption}{Assumption}

\newcommand{\eps}{{\epsilon}}

\newcommand{\E}{{\mathbb E}}

\newcommand{\M}{{\mathbb M}}

\newcommand{\R}{{\mathbb R}}
\newcommand{\N}{{\mathbb N}}
\renewcommand{\P}{{\mathbb P}}

\newcommand{\G}{{\mathbb{G}}}

\newcommand{\F}{{\mathcal F}}

\newcommand{\X}{{\mathcal{X}}}

\newcounter{rcnt}[section]

\def\argmin{\mathop{\rm argmin}}

\begin{document}
\onehalfspacing
\title{Parameter Estimation and Inference in a Continuous Piecewise Linear Regression Model}
\author{Georg Hahn, Moulinath Banerjee, and Bodhisattva Sen}
\date{Department of Statistics, Columbia University, New York, NY 10027, USA}
\maketitle

\begin{abstract}
The estimation of regression parameters in one dimensional broken stick models
is a research area of statistics with an extensive literature.
We are interested in extending such models by aiming to recover
two or more intersecting (hyper)planes in multiple dimensions.
In contrast to approaches aiming to recover a given number of piecewise linear components
using either a grid search or local smoothing around the change points,
we show how to use Nesterov smoothing to obtain a smooth and everywhere
differentiable approximation to a piecewise linear regression model with a uniform error bound.
The parameters of the smoothed approximation are then efficiently found by minimizing
a least squares objective function using a quasi-Newton algorithm.
Our main contribution is threefold:
We show that the estimates of the Nesterov smoothed approximation
of the broken plane model are also $\sqrt{n}$ consistent and asymptotically normal,
where $n$ is the number of data points on the two planes.
Moreover, we show that as the degree of smoothing goes to zero, the smoothed estimates
converge to the unsmoothed estimates and present an algorithm to perform parameter estimation.
We conclude by presenting simulation results on simulated data together with some guidance
on suitable parameter choices for practical applications.
\end{abstract}


\section{Introduction}
\label{section_introduction}
We are interested in parameter estimation and inference in a regression model of the type
\begin{equation}\label{eq:RegGen}
Y = g_\theta(X) + \epsilon,
\end{equation}
where $Y$ is the response variable, $X \in \X \subset \R^d$ ($d \ge 1$), and $g_\theta:\R^d \to \R$ is {\it continuous $k$-piecewise affine} ($k$-PWA; $k \ge 2$) --- there exists a partition of $\X$ into polyhedral sets $\{ C_i \}_{i=1}^k$ (i.e., $C_i \subset \X$, $C_i \cap C_j = \emptyset$ for all $i \neq j$) such that
\begin{equation}\label{eq:RegMdl}
	g_\theta(x) = a_i^\top x + b_i \qquad \mbox{if }\; x \in C_i,
\end{equation}
and that $g$ is continuous on $\X$; see e.g.,~\citet{Scholtes2012}. We assume that $k \in \N$ is given and denote the unknown parameter as $\theta = (a_1,\ldots, a_k,b_1,\ldots, b_k) \in \R^{k(d+1)}$. The unobserved error $\epsilon$ is assumed to have zero mean and finite variance.

Given i.i.d.~data $\{(X_i,Y_i)\}_{i=1}^n$ from the above model the goal is to estimate the unknown parameter $\theta$ and develop valid inferential procedures for the obtained estimator.  A naive approach to solving the above parametric regression problem is to consider the least squares estimator (LSE): 
\begin{equation}\label{eq:LSE-Naive}
\tilde \theta = \argmin_{\theta \in \R^{k(d+1)}} \sum_{i=1}^n (Y_i - g_\theta(X_i))^2,
\end{equation}
where the minimization is over all continuous PWA $g_\theta$. However, the above estimator is computationally intractable --- the optimization problem is non-smooth and non-convex; as noted in~\citet[Chapter 5]{Polyak1987}, virtually no computational guarantees are available for such problems. Such non-smooth functions can only be optimized using gradient-free or subgradient methods which typically attain a square root convergence rate as opposed to the superlinear convergence rate of the quasi-Newton method (under suitable conditions); see e.g.,~\cite{Shor1985} and~\cite{Nesterov2005}.

In this paper we resolve this non-smoothness in the estimation of $\theta$ by first smoothing $g_\theta$ appropriately and then minimizing the least squares criterion with the smoothed approximation of $g_\theta$ using a non-linear smooth optimization method such as \textit{BFGS} \citep{Broyden1970,Fletcher1970,Goldfarb1970,Shanno1970}. The novelty of our approach lies in the fact that: (i) we give provable bounds on the (smooth) approximation error of $g_\theta$, and (ii) we theoretically analyze the obtained computationally feasible estimator $\hat \theta$ and prove that $\hat \theta$ has the same statistical efficiency as $\tilde \theta$, the LSE described in~\eqref{eq:LSE-Naive}. 

Before we describe our procedure in detail let us look at two motivating real examples where modeling the regression function as in~\eqref{eq:RegMdl} can be useful.

\begin{table}[t]
\label{table_example2d}
\centering
\begin{tabular}{l||l|l||l|l}
&\multicolumn{2}{c||}{$d=1$}	&\multicolumn{2}{c}{$d=2$}\\
&average of $R$ 	&time [s]	&average of $R$ 	&time [s]\\
\hline
\cite{NelderMead1965}			&0.32	&75.5		&0.058	&1.9\\
Algorithm~\ref{algorithm_smoothing}	&0.30	&15.9		&0.047	&0.8
\end{tabular}
\caption{Empirical norm and computation time for the two segmented regressions in one (Example~\ref{example_1d}) and two (Example~\ref{example_2d}) dimensions. \cite{NelderMead1965} method and Algorithm~\ref{algorithm_smoothing}.}
\end{table}

\begin{figure}
\centering
\includegraphics[width=0.5\textwidth]{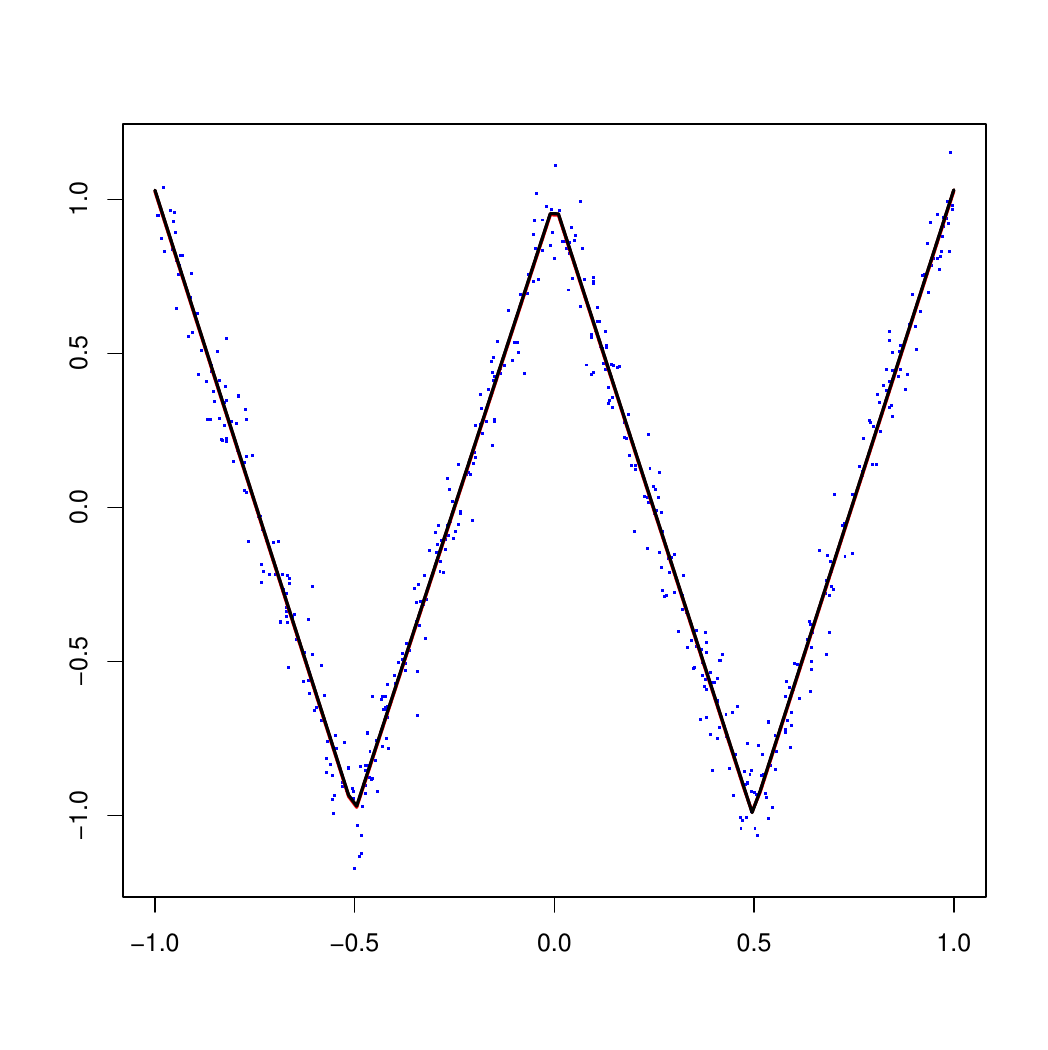}
\caption{Points generated from the PWA in \cite[Figure 2]{HempelEtAl13} (blue) and fitted PWA (red) consisting of a difference of two PWAs with three and two lines, respectively.
\label{fig:introexample1d}}
\end{figure}
\begin{figure}
\centering
\begin{minipage}{0.49\textwidth}
\includegraphics[width=\textwidth]{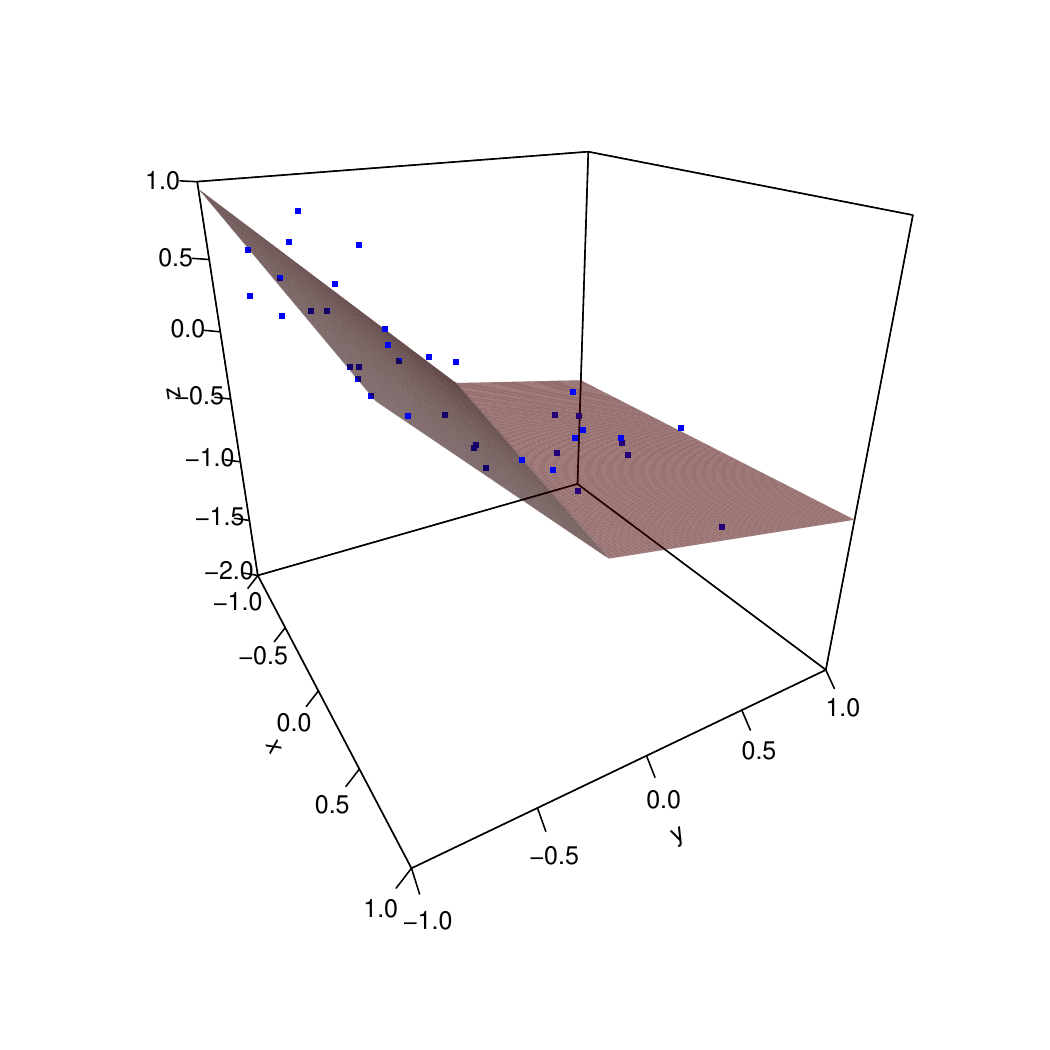}
\end{minipage}~
\begin{minipage}{0.49\textwidth}
\includegraphics[width=\textwidth]{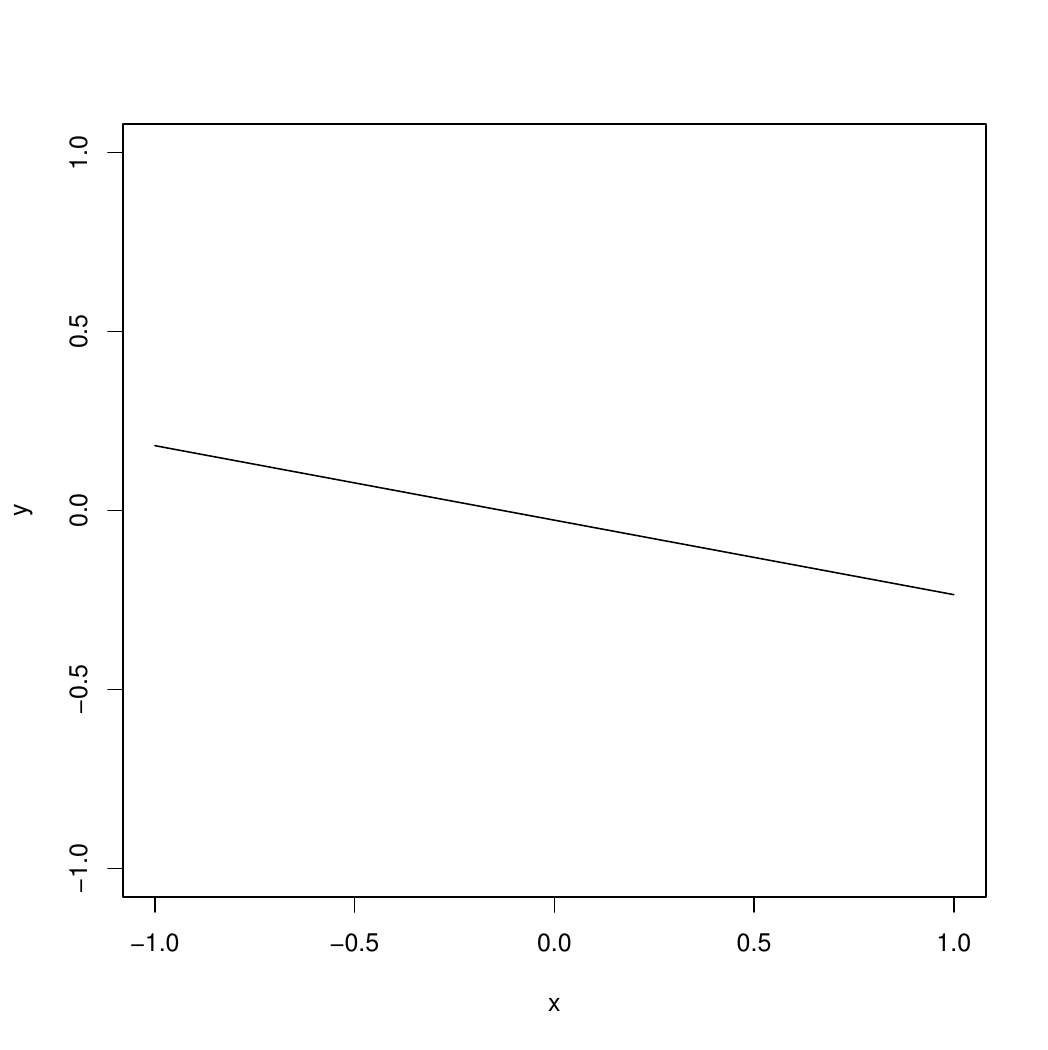}
\end{minipage}
\caption{Left: Visual representation of our two scenarios derived from the car dataset in \cite[Table 3]{HendersonVelleman1981}. 
Left: Variable \textit{MPG} as a function of the variables describing acceleration \textit{ACCEL} and weight \textit{WT}.
Right: Projected intersection line for the best fit displayed on the left.}
\label{fig:introexample2d}
\end{figure}

\begin{example}
\label{example_1d}
As a first example we consider the PWA in \cite[Figure 2]{HempelEtAl13} displayed in Figure~\ref{fig:introexample1d} (scaled down to lie in the interval $[-1,1]$ for conventional reasons).
This PWA is composed of a difference of two PWAs, the first one consisting of three and the second of two lines (see \cite{HempelEtAl13}).
We hence try to fit the PWA
$g_\theta(x) = \max_{i=1,2,3} \{ a_i^\top x+b_i \} - \max_{j=1,2} \{ c_i^\top x+d_j \}$,
where $\theta = (a_1,b_1,\ldots,a_3,b_3,c_1,d_1,c_2,d_2)$ and $a_i,b_i \in \R$, $c_i,d_i \in \R$ for $i=1,2,3$, $j=1,2$.
We first simulated $500$ points from this PWA by adding Gaussian noise (draws from iid.\ random variables with zero mean and standard deviation $0.1$) to uniformly selected function values.
We then use our approach (formalized as Algorithm~\ref{algorithm_smoothing} in Section~\ref{section_algorithm}) and the gradient-free \cite{NelderMead1965} algorithm to fit $g_\theta(x)$ to the data.

Algorithm~\ref{algorithm_smoothing} has one tuning parameter, the smoothness parameter $\mu$. In this and the following example we used $\mu=0.1$.
For the \cite{NelderMead1965} algorithm we used the implementation in the $R$ function \textit{optim}.
We measure the performance of any algorithm using the \textit{empirical norm} (the average squared residuals), that is $R:=\frac{1}{n} \sum_{i=1}^n (Y_i - g_{\hat{\theta}}(X_i))^2$, where $(X_i,Y_i)_{i=1}^n$ are the data, $g_\theta(x)$ is as above and $\hat{\theta}$ is the fitted parameter vector.

Results based on $1000$ repetitions are given in Table~\ref{table_example2d}. As seen from the table, both the \cite{NelderMead1965} and our algorithm perform comparably, yet Algorithm~\ref{algorithm_smoothing} yields results considerably faster even in the one dimensional case.
\end{example}

\begin{example}
\label{example_2d}
As a second example we consider the car dataset in \cite[Table 3]{HendersonVelleman1981} which has already been used in various publications to assess the performance of algorithms fitting breakpoint models.
This dataset contains consumer report data for $38$ car models captured in $11$ variables which include fuel consumption in miles per gallon, number of cylinders or the car weight.

We attempt to fit the fuel consumption \textit{MPG} as a function of the acceleration \textit{ACCEL} and weight \textit{WT} using a PWA consisting of two planes given by $g_\theta(x) = \max_{i=1,2} \{ a_i^\top x+b_i \}$, where $x,a_i \in \R^2$, $b_i \in \R$ for $i=1,2$. The choice of parameters for both algorithms as well as the evaluation using empirical norm and computational time are as in Example~\ref{example_1d}.

Table~\ref{table_example2d} (right) shows results. Our algorithm yields a speed-up of a factor two at a higher accuracy over the method of \cite{NelderMead1965} at comparable accuracy of both methods.
A graphical representation of the fitted functional is given in Figure~\ref{fig:introexample2d} (left).

Importantly, the best fit obtained with our algorithm displayed in Figure~\ref{fig:introexample2d} (left) finds two planes
with non-trivial intersection line shown in Figure~\ref{fig:introexample2d} (right) given by $y=-0.20 x-0.06$.
Approaches regressing only one variable at a time would thus not have been able to obtain a fit with comparable precision.
\end{example}

Many more real-life phenomena are captured by a model of type (\ref{eq:RegGen}) with continuous piecewise affine $g_\theta$:
these include, for instance, the Michigan Bone Health and Metabolism Study (MBHMS) in \cite{Das2015},
export prices in six EEC countries in \cite{Ginsburgh1980}, leukaemia and lung cancer data or time series of AIDS cases
\citep{StasinopoulosRigby1992,RigbyStasinopoulos1992,Molinari2001,Muggeo2003}.
\vspace{5mm}

We summarize the main contributions of the paper below.
\begin{enumerate}
  \item Our smoothing approach uses ideas from~\cite{Nesterov2005} and yields a smooth approximation $g^{\mu}_\theta$ of $g_\theta$ that is uniformly close (up to any desired precision depending on the smoothing parameter $\mu >0$) to $g_\theta$. In fact, we can provide theoretical bounds on the quality of the smooth approximation (i.e., $|g_\theta^{\mu} - g_\theta|$). Our approach is very different from the usual techniques (e.g., kernel smoothing) employed in nonparametric
statistics to obtain smooth approximations of non-smooth functions. To apply our smoothing technique we first express $g_\theta$ as the difference of two PWA convex functions which are then smoothed separately while preserving convexity.  

\item We give an algorithm based on a quasi-Newton method to estimate $\theta$ which is fast and stable in practice. We evaluate our procedure  on simulated data and provide empirical guidelines on how to choose the tuning parameter $\mu$ in practice. In contrast to existing approaches summarized below, which usually address $d=1$ or $2$, the results in this article are not limited to any dimension. Moreover, existing methods such as the ones of \cite{TishlerZang1981} or \cite{Muggeo2003} which are able to handle more than $k=2$ components are usually based on heuristics and lack a proper theoretical study. To the best of our knowledge, the proposed procedure seems to be the first attempt at a systematic study of the estimation of $\theta$ under such generality (i.e., $k \ge 2$ and $d \ge 1$).
  
\item To describe the theoretical results on the statistical performance of the proposed estimator $\hat \theta$ let us consider the case $k=2$, which has received some attention in the statistical literature --- see e.g.,~\cite{Ginsburgh1980}, \cite{SmithCook1980}, \cite{BaconWatts1971}, \cite{vdGeer1988}. In this situation, $g_\theta$ is either convex or concave. Without loss of generality let us assume that $g_\theta$ is convex, in which case $g_\theta$ can be easily represented as
\begin{align}
g_\theta(x) = \max \{ \alpha^\top x+\gamma,\beta^\top x+\phi\},
\label{eq:g}
\end{align}
where $\theta=(\alpha,\beta,\gamma,\phi) \in \R^{2d+2}$. Even in this simple setup, the computation of $\tilde \theta$ (using~\eqref{eq:LSE-Naive}) is non-trivial:
The non-smoothness of $g_\theta$ complicates the use of gradient based optimization methods (such as quasi-Newton or Newton-Raphson) to minimize the least squares criterion. In this setting we theoretically analyze the statistical properties of our proposed estimator $\hat \theta$. We show that, under proper choices of $\mu$ and mild conditions on the distribution of $(X,Y)$, $\hat \theta$ is $\sqrt{n}$-consistent and asymptotically normal with the {\it same} limiting distribution as that of the computationally infeasible LSE $\tilde \theta$. This immediately yields confidence intervals for $\theta$. Our analysis, in principle, can be extended beyond the setting $k=2$ where similar results would also hold for the proposed estimator.
\end{enumerate}

Early work on linear regression with change points dates back to~\cite{Quandt1958} and \cite{Quandt1960} who consider $d=1$; also see~\cite{Blischke1961},~\cite{Robison1964},~\cite{Hudson1966}. Asymptotic results on the limiting distributions of the LSE in regression models with separate analytical forms in different regions appear in~\cite{Feder1975}. \cite{SiegmundZhang1994}
develop conservative confidence regions for the change point of a broken stick model (i.e., $k=2$) in one dimension.

The idea of smoothing dates back to \cite{TishlerZang1981}.
In order to overcome the lack of smoothness around each change point,
they propose to use a smooth quadratic approximation of the model in an interval $(-\beta,\beta)$
around each change point, hence allowing for efficient minimization of the likelihood function using the Newton algorithm.
Although both the smooth approximation as well as the dependence on the parameter $\beta$
are evaluated numerically, no theoretical results or asymptotic distributions are given. The present work is closely related to the work of \cite{Das2015}
who consider a broken stick model in two dimensions, possibly with multiple change points.
\cite{Das2015} propose a method relying on local smoothing in a neighborhood of each change point
and show that the resulting estimate is $\sqrt{n}$ consistent and asymptotically normal.

Nevertheless, all of the aforementioned works suffer from the drawback that
only one dimensional problems are considered, and that their results are not extendable to higher dimensions.

The article is structured as follows.
Section~\ref{section_smoothed} discusses the smoothing of piecewise affine functions,
presents the \cite{Nesterov2005} smoothing technique (Section~\ref{subsection_Nesterov_smoothing})
and discusses the choice of the so-called prox function required for smoothing (Section~\ref{subsection_prox}).
We investigate two smoothing techniques based on two different prox functions.
An algorithm to compute estimates of both the change point as well as the regression parameters
will then be given in Section~\ref{section_algorithm}.
Section~\ref{section_asymptotic_results} discusses asymptotic results.
We show that the least squares estimate of our proposed smoothed problem is
also $\sqrt{n}$ consistent, asymptotically normal and moreover converges to the one of the unsmoothed
problem as the \cite{Nesterov2005} smoothing parameter goes to zero at an appropriate rate.
Section~\ref{section_experimentals} presents selected simulation results demonstrating
the accuracy of the estimates obtained with our approach,
verifying convergence rates and highlighting computational issues.
All proofs can be found in the appendix.

\section{Smoothing piecewise affine functions}
\label{section_smoothed}
Consider the regression model~\eqref{eq:RegGen} where $g_\theta$ is a continuous piecewise affine (PWA) function. The following result, which states that any continuous PWA function can be expressed as the difference of two convex PWA functions, will be crucial for the rest of the sequel.

\begin{lemma}\label{lem:BasicDecomp}
Let $\X \subset \R^d$ be a convex polyhedral region (i.e., a set formed by the intersection of finitely many hyperplanes). Every continuous PWA function $g: \X \to \R$ defined over a convex polyhedral partition of $\X$ with full dimensional elements $C_i$, $i=1,\ldots, k$, can be written as the difference of two convex PWA functions, i.e., $$g(x) = g_1(x) - g_2(x), \qquad \mbox{ for all } x \in \X,$$ where $g_i: \X \to \R$, $i=1,2$, are convex and PWA functions.
\end{lemma}
The above result is proved in~\cite{KS87}; also see~\cite{HempelEtAl13}. Consequently, to find a smooth approximation of a PWA (continuous) $g$ we first consider the case when $g$ is PWA and also {\it convex}. The advantage of using a convex PWA is that it can be conveniently represented as a maximum of affine functions, i.e., 
\begin{equation}\label{eq:Cvx-PWA}
g(x) = \max_{i=1,\ldots, k} \{a_i ^\top x + b_i\}
\end{equation} 
for some $k \ge 1$ where $a_1,\ldots, a_k \in \R^d$ and $b_1,\ldots, b_k \in \R$. Let $A$ be a $k \times (d+1)$ matrix such that the $i$-th row of $A$ is $(a_i,b_i) \in \R^{d+1}$. Then $g$ can be succinctly represented as $$g(x) = \max_{i=1,\ldots, k} (A [x,1])_i,$$ where for a vector $u\in \R^k$, $(u)_i$ denotes the $i$-th component of $u$, and where $[x,1] \in \R^{d+1}$ is the concatenation of vector $x$ and scalar $1$. 

In the following subsection we detail an approach to smooth a convex PWA. This procedure will then be used to construct a smooth approximation to any continuous PWA function.

\subsection{The smoothing approach}
\label{subsection_Nesterov_smoothing}
Suppose that $f:\R^q \to \R$ ($q \ge 1$) be a piecewise linear convex function. Then we can express $f$ as
\begin{equation}\label{eq:Represent-CvxPWA}
f(z) = \max_{i=1,\ldots, p} (A z)_i, \qquad \mbox{for all } z \in \R^q,
\end{equation} for some matrix $A$ of order $p \times q$. In light of the specific representation of interest given in \eqref{eq:Cvx-PWA}, the dimension $q$ will correspond to $q=d+1$ and likewise for a $k$-PWA, we will later consider $p=k$. In this subsection we find a smooth convex approximation to any $f$ of the form defined in~\eqref{eq:Represent-CvxPWA} and study some of its properties.

Let $\|\cdot \|_p$ ($p \ge 1$) be a norm on $\R^p$ and let $\langle \cdot,\cdot \rangle$ denote the usual Euclidean inner product. Let $Q_p \subset \R^p$ be the unit simplex in $\R^p$, i.e., $$Q_p := \left\{w = (w_1,\ldots, w_p) \in \R^p: \sum_{i=1}^p w_i = 1, \mbox{ and } w_i \ge 0, \mbox{ for all } i=1,\ldots, p \right\}.$$ Following \cite{Nesterov2005}, let $\rho$ be a proximity function (\textit{prox function}) --- a nonnegative continuously differentiable {\it strongly convex} function (with respect to the norm $\|\cdot\|_p$) on $Q_p$, i.e., 
$$\rho(s) \geq \rho(t) + \langle \nabla \rho(t),t-s \rangle + \frac{1}{2} \Vert t-s \Vert_p, \qquad \mbox{for all }  s,t \in Q_p.$$ Consider now the function $f^\mu: \R^q \rightarrow \R$ defined as 
\begin{equation}
f^\mu(z) := \max_{w \in Q_p} \left\{ \langle Az,w \rangle - \mu \rho(w) \right\},
\label{eq:nesterov}
\end{equation}
where $\mu >0$ is a tuning parameter. Then $f^\mu$ is our {\it smooth approximation} of $f$. Observe that when $\mu = 0$, $f^0(z)$ recovers the unperturbed function $f$ (cf.~\eqref{eq:Represent-CvxPWA}): $$f(z) = \max_{w \in Q_p} \left\{ \langle Az,w \rangle \right\} = f^0(z).$$
The following lemma, taken from~\citet[Theorem 1]{Nesterov2005}, shows that $f^\mu$ is a smooth convex function.
\begin{lemma}
\label{lemma_nesterov}
For any $\mu > 0$,
the function $f^\mu$ defined in (\ref{eq:nesterov}) is convex and everywhere differentiable in $z$.
The gradient of $f^\mu$ is given by $\frac{\partial}{\partial z}f^\mu(z) = A^\top \hat{w}_\mu$, where
$\hat{w}_\mu = \arg\max_{w \in Q_p} \left\{ \langle Az,w \rangle - \mu \rho(w) \right\}$.
Moreover, the gradient $z \mapsto \frac{\partial}{\partial z} f^\mu(z)$ is Lipschitz continuous with parameter $\Vert A \Vert_{p,q}^2/\mu$, where $$\Vert A \Vert_{p,q} := \max_{u,v} \left\{ \langle Au,v \rangle: \Vert u \Vert_q = 1, \Vert v \Vert_p = 1, u \in \R^q, v \in \R^p \right\}$$ is a norm on the space of matrices in $\R^{p \times q}$.
\end{lemma}

The definition of the smoothed estimator $f^\mu$ in \eqref{eq:nesterov} immediately allows us to obtain bounds on the approximation error $|f(z) - f^\mu(z)|$. Indeed,
\begin{align}
f^\mu(z) \geq \sup_{w \in Q_p} \langle Az,w \rangle - \mu \sup_{w \in Q_p} \rho(w) = f^0(z) - \mu \sup_{w \in Q_p} \rho(w), \nonumber\\
f^\mu(z) = \sup_{w \in Q_p} \left\{ \langle Az,w \rangle - \mu \rho(w) \right\} \leq \sup_{w \in Q_p} \langle Az,w \rangle = f^0(z), \nonumber
\end{align}
using the nonnegativity of the prox function \citep{Nesterov2005}; also see~\citet[Section 3.1]{Mazumder2015} for a detailed description of this approach. Summarizing the above considerations we obtain the following two useful results:
\begin{enumerate}
  \item The function $f^\mu$ defined in (\ref{eq:nesterov}) is smooth for any $\mu>0$, convex and has a Lipschitz continuous gradient which is proportional to $\mu^{-1}$. The original unsmoothed PWA function $f$ is recovered by setting $\mu=0$.
  \item $f^\mu(z)$ is a {\it uniform approximation} to $f^0(z)$ as
\begin{equation}
  f^0(z) - \mu \sup_{w \in Q_p} \rho(w) \leq f^\mu(z) \leq f^0(z).
  \label{eq:Nesterovbounds}
  \end{equation}
Thus the uniform approximation error is upper bounded by $$\sup_{z \in \R^p}|f(z) - f^\mu(z)| \le \mu \sup_{w \in Q_p} \rho(w) = O(\mu)$$ which depends only on $\mu$ and $\rho$.
\end{enumerate}
The choice of $\mu$ will be important and we give some sufficient conditions on $\mu$ in Section~\ref{section_asymptotic_results} for our main theoretical results to hold. We further provide practical guidelines on how to choose $\mu$ in Section~\ref{subsection_dependence_smoothing_parameter}. There are various possible choices of the prox function $\rho$. We will discuss two (standard) choices in detail in the following section which are also used throughout the paper, including in our simulation studies.

We aim to use the approach in~\eqref{eq:nesterov} to obtain a smooth uniform approximation to any convex PWA $f$ of the form defined in~\eqref{eq:Cvx-PWA}.

\subsection{Choice of the prox function}
\label{subsection_prox}
We consider two prox functions in this article, one based on the squared error loss and the other based on the entropy loss. As will be shown in Section~\ref{subsection_smoothed} both these prox functions yield $\sqrt{n}$-consistent and asymptotically normal estimators of $\theta$ with the {\it same} limiting distribution as the one of the naive LSE $\tilde \theta$.

\begin{itemize}
  \item \noindent {\bf Entropy prox function:} The entropy prox function $\rho: \R^p \rightarrow \R$ is
given by $$\rho(w) = \sum_{i=1}^p w_i \log(w_i)+\log p,$$ where $w = (w_1,\ldots, w_p)$. Letting $\Vert \cdot \Vert_p$ be the $\ell_1$-norm in $\R^p$ (i.e., $\|u\|_{p} = \sum_{i=1}^p |u_i|$) it can be shown that $\rho$ is strongly convex with respect to this norm and satisfies $$\sup_{w \in Q_p} \rho(w) = \log p;$$
see e.g.,~\cite{Nesterov2005}.
More importantly, the entropy prox function allows an analytic expression: for the PWA convex function $f$ given in~\eqref{eq:Cvx-PWA}, the prox-smoothed approximation of $f$ is given by
\begin{align}
\nonumber
f^\mu(x) &= \max_{w \in Q_k} \left\{
\sum_{i=1}^k w_i (a_i^\top x+b_i) - \mu \left( \sum_{i=1}^k w_i \log w_i  + \log k  \right) \right\}\\
&= \mu \log \left( \frac{1}{k} \sum_{i=1}^k e^\frac{a_i^\top x+b_i}{\mu} \right).
\label{eq:formular_entropy}
\end{align}
Further, using~\eqref{eq:Nesterovbounds}, we have
\begin{equation}\label{eq:EntProxApp}
\sup_{x \in \R^d} |f(x)-f^\mu(x)| \leq \mu \sup_{w \in  Q_k} \rho(w) = \mu \log k.
\end{equation}

\item \noindent {\bf Squared error prox function:} The squared error prox function $\rho: \R^p \rightarrow \R$ is
given by $$\rho(w) = \frac{1}{2} \left\Vert w - \frac{1}{m} \mathbf{1} \right\Vert_p^2,$$ where $\mathbf{1} \in \R^d$ is the vector of ones and $\Vert \cdot \Vert_p$ denotes the Euclidean norm in $\R^p$. Following \cite{Mazumder2015}, the optimization problem in (\ref{eq:nesterov}) applied to the PWA function in~\eqref{eq:Cvx-PWA} is equivalent to the following convex program:
\begin{align}
f^\mu(x) =  \min_{w \in Q_k} \left\{ \frac{1}{k} \sum_{i=1}^k w_i^2 - \sum_{i=1}^k w_i c_{i}^{\theta,\mu}(x) \right\},
\label{eq:michelot}
\end{align}
where $c_{i}^{\theta,\mu}(x) := (a_i^\top x + b_i)/\mu - 1/k$, $i \in \{ 1,\ldots,k \}$, and $\theta = (a_1,\ldots,a_k,b_1,\ldots,b_k)$. This problem is identical to the one of finding the Euclidean projection of the vector $(c_{i}^{\theta,\mu}(x))_{i=1}^k$ onto $Q_k$, the $k$-dimensional unit simplex. It can be solved efficiently using the algorithm of \cite{Michelot1986}. Denoting the Euclidean projection of the vector $(c_{i}^{\theta,\mu}(x))_{i=1}^k$ onto $Q_k$ by $\hat w^{\theta,\mu}(x) = (\hat w_{i}^{\theta,\mu}(x))_{i=1}^k$, the prox-smoothed approximation for the PWA convex function $f$ given in~\eqref{eq:Cvx-PWA} can be written as
\begin{align}
f^\mu(x) = \sum_{i=1}^k \hat w_{i}^{\theta,\mu}(x) \cdot (a_i^\top x+b_i) - \mu \rho(\hat w^{\theta,\mu}(x)).
\label{eq:formular_sq}
\end{align}
As $\sup_{w \in Q_k} \frac{1}{k} \left| w - \frac{1}{k} \mathbf{1} \right|^2 = 1-\frac{1}{k}$, we have
\begin{equation}\label{eq:SqErrProxApp}
\sup_{x \in \R^d} |f(x)-f^\mu(x)| \leq \mu \sup_{w \in Q_k} \rho(w) = \mu \left( 1-\frac{1}{k} \right).
\end{equation}
\end{itemize}

\section{Our Algorithm}
\label{section_algorithm}
In this section we describe our algorithm to estimate the $k$ hyperplanes given in~\eqref{eq:RegMdl}. By Lemma~\ref{lem:BasicDecomp} we know that the continuous $k$-PWA function $g_\theta$ (in~\eqref{eq:RegMdl}) can be represented as
\begin{equation}\label{eq:DiffOf2Cvx}
g_\theta(x) = g_{1,\theta_1}(x) - g_{2,\theta_2}(x)\qquad \mbox{for all } x \in \X,
\end{equation}
where $\theta=(\theta_1,\theta_2)$ and $g_{i,\theta_i}$, $i=1,2$, is a convex PWA with the representation (see~\eqref{eq:Cvx-PWA})
\begin{equation}\label{eq:EachCvxFnc}
g_{i,\theta_i}(x) = \max_{j=1,\ldots, k_i} \{a_{i,j} ^\top x + b_{i,j}\}
\end{equation}
for $\theta_i=(a_{i,1},\ldots,a_{i,k_i},b_{i,1},\ldots,b_{i,k_i}) \in \R^{k_i(d+1)}$ and nonnegative integers $k_i$ (note that $k_i = \# \theta_i/(d+1)$, where $\# \theta_i$ denotes the length of vector $\theta_i$). We assume here that $k_1$ and $k_2$ are specified in advance by the user. For identifiability reasons we can take $a_{2,1}=0, b_{2,1}=0$ as formalized in the next lemma.
\begin{lemma}
\label{lemma_identifiability}
In the model~\eqref{eq:DiffOf2Cvx} with $g_{i,\theta_i}$ as defined in~\eqref{eq:EachCvxFnc}, if $k_i \ge 1$ for $i=1,2$, we can without loss of generality assume that $a_{2,1}=0, b_{2,1}=0$.
\end{lemma}
\begin{proof}
We have $g_\theta(x) = \max_{j=1,\ldots, k_1} \{a_{1,j} ^\top x + b_{1,j}\} - \max_{j=1,\ldots, k_2} \{a_{2,j} ^\top x + b_{2,j}\}$. We can express $g_\theta$ as:
$$g_\theta(x) = \max_{j=1,\ldots, k_1} \{a_{1,j} ^\top x + b_{1,j}\} - \max_{j=1,\ldots, k_2} \{a_{2,j} ^\top x + b_{2,j}\}
\pm (a_{2,1} ^\top x + b_{2,1})$$
and use the fact that $\max \{ u,v \} - w = \max \{ u-w,v-w \}$ for arbitrary $w,v,w \in \R$.
Setting $\bar{a}_{i,j} := a_{i,j} - a_{2,1}$ and $\bar{b}_{i,j} := b_{i,j}-b_{2,1}$ for all $j=1,\ldots,k_i$, $i =1, 2$, leads to
$$g_\theta(x) = \max_{j=1,\ldots, k_1} \{\bar{a}_{1,j} ^\top x + \bar{b}_{1,j}\} - \max_{j=1,\ldots, k_2} \{\bar{a}_{2,j} ^\top x + \bar{b}_{2,j}\}$$ with $\bar{a}_{2,1}=0, \bar{b}_{2,1}=0$.
\end{proof}
With this change in parametrization, the goal is now to estimate the parameter $\theta = (\theta_1,\theta_2) \in \R^{(k_1 + k_2)(d+1)}$. Given a set of data points $\{(X_i,Y_i)\}_{i=1}^n$ we use the method of least squares to estimate $\theta$. As $g_{\theta}$ is non-smooth we cannot directly employ a gradient descent algorithm. 

As a first step to resolve this difficulty, we compute smoothed convex approximations to $g_{1,\theta_1}$ and $g_{2,\theta_2}$ using the approach outlined in Section~\ref{subsection_Nesterov_smoothing} and Section~\ref{subsection_prox}. This leads to smoothed functions $g_{1,\theta_1}^\mu $ and $g_{2,\theta_2}^\mu $ and thus to a smoothed approximation of $g_\theta$: $$g^\mu_\theta := g_{1,\theta_1}^\mu -g_{2,\theta_2}^\mu.$$
Let $$M_n^\mu(\theta) :=\frac{1}{n} \sum_{i=1}^n \left( Y_i - g^\mu_\theta(X_i) \right)^2$$ be the least squares criterion function we now try to minimize over $\theta \in \R^{(k_1+k_2)(d+1)}$. Thus, the LSE of $\theta$ we consider is 
\begin{equation}\label{eq:SmoothedLSE}
\hat \theta = (\hat \theta_1,\hat \theta_2) := \argmin_{\theta \in \R^{(k_1+k_2)(d+1)}}M_n^\mu(\theta).
\end{equation}
We study the computation and the statistical properties of  $\hat \theta$ is this paper. In this section we focus on the computation of $\hat \theta$. To minimize $M_n^\mu$ we first compute its Jacobian matrix:
$$J_n^\mu(\theta) = 2 \sum_{i=1}^n \left( Y_i-g_{\theta}^\mu(X_i) \right) \cdot \nabla_\theta g_\theta^\mu(X_i),$$
where the gradient $\nabla_\theta g_\theta^\mu$ varies depending on the choice of the prox function used for smoothing. As $$\nabla_\theta g_\theta^\mu(x) = [\nabla_{\theta_1} g_{1,\theta_1}^\mu(x), - \nabla_{\theta_2} g_{2,\theta_2}^\mu(x)] \qquad \mbox{for all }\; x \in \X,$$ we have to first find the gradients of the function $g_{i,\theta_i}^\mu$ ($i=1,2$) with respect to $\theta_i$. To this end we state the general form of the gradients for \cite{Nesterov2005} smoothed functions.
\begin{lemma} \label{lemma_gradient}
For $\theta = (a_1,\ldots,a_{k},b_1,\ldots,b_{k}) \in \R^{k(d+1)}$, let $$f_\theta(x) := \max_{j=1,\ldots, k} \{a_j ^\top x + b_j\} = \max_{j=1,\ldots, k} (A [x,1])_j,$$
where $A \in \R^{k \times (d+1)}$ is a matrix whose $j$'th row is given by $(a_j,b_j) \in \R^{d+1}$. Then the derivatives (with respect to $a_j$'s and $b_j$'s) of the \cite{Nesterov2005} smooth approximation of $f_\theta(x)$ defined in~\eqref{eq:nesterov}, i.e.\
\begin{equation}\label{eq:MaxQ_p}
	f^\mu_\theta(x) := \max_{w \in Q_p} \left\{ \langle A[x,1],w \rangle - \mu \rho(w) \right\} = (A[x,1])^\top \hat{w}^{\theta,\mu} - \mu \rho(\hat{w}^{\theta,\mu})
\end{equation} 
where $\hat{w}^{\theta,\mu} := \arg\max_{w \in Q_p} \left\{ \langle A[x,1],w \rangle - \mu \rho(w) \right\}$  (as in Lemma~\ref{lemma_nesterov})
are $$\frac{\partial f^\mu_{\theta}}{\partial a_j}(x) = \hat w_j^{\theta,\mu} x \qquad \mbox{and} \qquad \frac{\partial f^\mu_\theta}{\partial b_j}(x) = \hat w_j^{\theta,\mu}\qquad \mbox{for } j =1,\ldots, k.$$
\end{lemma}
\begin{proof}
Fix $x \in \R^d$. Indeed, $f^\mu_\theta(x)$ (viewed as a function of $\theta$) is a maximum of functions, which are linear in $a_j$'s and $b_j$'s. Note that $f^\mu_\theta(x)$ admits the representation $$f_\theta^\mu(x) =  \sum_{j=1}^k \hat w_j^{\theta,\mu} (a_j^\top x) + \sum_{j=1}^k \hat w_j^{\theta,\mu} b_j - \mu \rho(\hat{w}^{\theta,\mu}).$$ The result now follows because $f^\mu_\theta(x)$ is differentiable in $\theta$ since $\hat w^{\theta,\mu}$ is unique (note that $\langle A[x,1],w \rangle - \mu \rho(w)$ is a strongly convex function in $w$).
\end{proof}

We now specialize in the two prox functions used before and give explicit expressions for $\nabla_{\theta_i} g_{i,\theta_i}^\mu(x)$, for $i = 1,2$.

\begin{itemize}
 \item {\bf Squared Error prox function:}
The approximation $g_{i,\theta_i}^\mu$ with squared error prox function makes use of the projected vector $\hat{w}^{\theta,\mu}$ maximizing~\eqref{eq:MaxQ_p} (see the discussion after~\eqref{eq:michelot}). It is thus straightforward to apply Lemma~\ref{lemma_gradient}, leading to
$$\nabla_{\theta_i} g_{i,\theta_i}^\mu(x) = [x,1] \otimes \hat{w}^{\theta_i,\mu},$$ where
for two vectors $u=(u_1,\ldots,u_p) \in \R^p$ and $v=(v_1,\ldots,v_q) \in \R^q$
we define
$u \otimes v := \left( u_1 v_1,\ldots,u_p v_1, u_1 v_2, \ldots, u_p v_2,u_1 v_3, \ldots \right)$
and
$\hat{w}^{\theta,\mu}$ is as defined in Lemma~\ref{lemma_gradient}.

\item {\bf Entropy prox function:}
Although Lemma~\ref{lemma_gradient} gives an explicit derivative of each $g_{i,\theta_i}^\mu$ for $i=1,2$, the vector $\hat{w}^{\theta,\mu} \in Q_p$ maximizing~\eqref{eq:MaxQ_p} is non-trivial to compute.
Instead, the closed form expression
$$g_{i,\theta_i}^\mu(x) = \mu \log \left( \frac{1}{k_i} \sum_{j=1}^{k_i} e^\frac{a_{i,j}^\top x+b_{i,j}}{\mu} \right)$$
of the smooth approximation of $g_{i,\theta_i}$ given in (\ref{eq:formular_entropy}) can be differentiated directly.
This leads to $$\nabla_{\theta_i} g_{i,\theta_i}^\mu(x) = \Vert r_{i,\theta_i} \Vert_{k_i}^{-1} \cdot \left[ x \otimes r_{i,\theta_i}, r_{i,\theta_i} \right],$$
where $r_{i,\theta_i} = \left( e^\frac{a_{i,1}^\top x+b_{i,1}}{\mu}, \ldots, e^\frac{a_{i,k_i}^\top x+b_{i,k_i}}{\mu} \right)$
and $\Vert \cdot \Vert_{k_i}$ is the $l_1$-norm in $\R^{k_i}$.
\end{itemize}

Once the computation of the Jacobian matrix is completed, we use a hill-climbing optimization technique (quasi-Newton method) with random initial starting value to minimize $M_n^\mu(\theta)$. Details are given in Algorithm~\ref{algorithm_smoothing}.
\begin{algorithm}
\label{algorithm_smoothing}
\caption{\texttt{Computation of the LSE for the smoothed PWA function}}
\SetKwInOut{Input}{input}
\SetKwInOut{Output}{output}
\SetKwFor{Loop}{repeat}{}{end}
\Input{data points $\{ (X_i,Y_i) \}_{i=1}^n$, number of planes $k_1,k_2 \in \N$, smoothing parameter $\mu>0$, tolerance $\tau>0$ for convergence of quasi-Newton method}
\Output{estimate $\hat{\theta}$}
Determine $m_0 \in \N$ such that $2^{m_0} \mu > 1$\;
Sample random initial starting value $\hat{\theta}_0=(a_{1,1},\ldots,a_{1,k_1},b_{1,1},\ldots,b_{1,k_1},a_{2,1},\ldots) \in [-r,r]^{(k_1+k_2)(d+1)}$ for some $r>0$\;
\label{algo:initial}
\For{$m \leftarrow 0$ \KwTo $m_0$}{
Set $\mu_m := 2^{m_0-m} \mu$\;
Perform quasi-Newton minimization of $M_n^{\mu_m}(\theta)$ with initial starting value $\hat{\theta}_m$, gradient $J_n^{\mu_m}(\theta)$ and tolerance $\tau$; if the quasi-Newton method fails to converge (within a pre-set number of steps) then restart from line~\ref{algo:initial}\;
\label{algo:minimization}
Set minimum found by quasi-Newton step as $\hat{\theta}_{m+1}$\;
}
\Return{$\hat{\theta}:=\hat{\theta}_{m_0+1}$};
\end{algorithm}

The idea behind Algorithm~\ref{algorithm_smoothing} is as follows.
Due to the fact that the degree of smoothness of
$g_{\theta}^\mu$ (and hence of $M_n^\mu(\theta)$) decreases as $\mu$ vanishes, minimizing $M_n^\mu(\theta)$ becomes increasingly challenging as $\mu \to 0$. To overcome this problem, we propose to iteratively refine the least squares solution by starting with a large initial value for the smoothness parameter ($\mu_0>1$ in Algorithm~\ref{algorithm_smoothing}; the value $1$ is an arbitrary choice) and by decreasing the smoothness parameter by a factor of two in every iteration. For this we first determine a $m_0 \in \N$ such that $2^{m_0} \mu > 1$, thus making sure that after $m_0$ iterations, an estimate of $\theta$ for the desired value $\mu$, chosen by the user, is obtained. The optimization itself is carried out using a standard quasi-Newton scheme such as \textit{BFGS} \citep{Broyden1970,Fletcher1970,Goldfarb1970,Shanno1970}. In each iteration $m$,
the estimate $\hat{\theta}_m$ from the last iteration serves as a new initial value for the next call of the minimization (quasi-Newton) method in line~\ref{algo:minimization} with decreased value of $\mu$ and a user-specified tolerance $\tau$ (typically of the order of $10^{-5}$).
Alternative criteria for termination can also be employed.
If any instance of the quasi-Newton method fails to converge or shows numerical instabilities
(for instance due to singularity of the Jacobian matrix), the whole algorithm is restarted using a new random initial value.
The estimate $\hat{\theta}$ returned by Algorithm~\ref{algorithm_smoothing}
(that is the estimate corresponding to the desired degree of smoothness $\mu$) is the minimum found in last minimization (line~\ref{algo:minimization}) for $\mu_{m_0}=\mu$.

\section{Statistical properties of the LSE: asymptotic normality and efficiency}
\label{section_asymptotic_results}

\subsection{The unsmoothed case}
\label{subsection_unsmoothed}
In this section we study the statistical properties of the LSE (obtained from the smoothed PWA function) discussed in Section~\ref{section_algorithm}. To keep the presentation simple and the technical arguments less cumbersome, we  consider the case when $k_1=2$ and $k_2 = 0$, i.e., the regression function $g_\theta$ can be expressed as
\begin{equation}\label{eq:g_Cvx}
g_\theta(x) = \max\{\alpha^\top x + \gamma, \beta^\top x+ \phi\} \quad \mbox{for }
\theta = (\alpha,\beta, \gamma, \phi) \in \R^{2d+2}.
\end{equation}
Let $\{(X_i,Y_i)\}_{i=1}^n$ be i.i.d.\ (having joint distribution $P$ on $\R^d \times \R$) from the regression model
$$Y = g_\theta(X) + \eps,$$
where $Y$ is the response variable, $X \in \X \subset \R^d$ is the predictor, $\eps$ is the unobserved error such that $\E(\eps|X) = 0$ almost everywhere (a.e.) and has finite variance $\sigma^2$, and $g_\theta$ has the form given in~\eqref{eq:g_Cvx}. We assume that the unknown parameter $\theta = (\alpha,\beta, \gamma, \phi) \in \Theta$, where $\Theta$ is a compact set in $\R^{2d+2}$.

Let $\theta_0$ be the true value of the parameter $\theta$, which we assume lies in the interior of $\Theta$. Our goal is to estimate $\theta_0$ from the observed data by using the method of least squares:
\begin{align}
\hat{\theta}_n := \arg \min_{\theta \in \Theta} \frac{1}{n} \sum_{i=1}^n [Y_i - g_\theta(X_i)]^2.
\label{eq:P}
\end{align}

Computing the least squares estimate (\ref{eq:P}) is challenging in practice.
This is due to the lack of smoothness of the functions in the class $\mathcal{G} := \{g_\theta: \theta \in \Theta\}$ under consideration. To remedy this problem,
the following subsection investigates the smoothed approximation of the class $\mathcal{G}$ of \cite{Nesterov2005}
which allows to compute a gradient for any $g_\theta \in \mathcal{G}$ and hence the use of quasi-Newton methods to compute
the least squares estimate (\ref{eq:P}).

We first aim to establish that the least squares estimator of $\theta_0$ for the unsmoothed
class of functions $\mathcal{G}$ is $\sqrt{n}$ consistent and asymptotically normal.
This result is summarized in the next lemma and relies on the following assumption.

\begin{assumption}[Moment conditions]
\label{as:1}
$\E[\|X\|] < \infty, \;  \E[|Y|\|X\|] < \infty,\; \E[Y^2] < \infty.$
We further assume that $X$ does not put all its mass on any hyperplane in $\R^d$. 
\end{assumption}

\begin{lemma}
\label{lemma_unsmoothed}
Under Assumption~\ref{as:1}, the following statements hold true.
\begin{enumerate}
  \item $\hat{\theta}_n-\theta_0 = O_P(n^{-1/2})$.
  \item As $n \rightarrow \infty$,
$\sqrt{n}(\hat{\theta}_n - \theta_0)$ converges to a multivariate $\text{Normal}(0,V^{-1} W V^{-1})$ distribution, where
\begin{align*}
W &= 4 \sigma^2 \int\limits_{g_{\theta_0}(x)=\alpha_0^\top x + \gamma_0}
\begin{pmatrix} x x^\top & x & 0 & 0\\ x^\top & 1 & 0 & 0\\ 0&0&0&0\\ 0&0&0&0 \end{pmatrix}dP(x)\\
&+ 4 \sigma^2 \int\limits_{g_{\theta_0}(x)=\beta_0^\top x + \phi_0}
\begin{pmatrix} 0&0&0&0\\ 0&0&0&0\\ 0 & 0 & x x^\top & x\\ 0 & 0 & x^\top & 1 \end{pmatrix}dP(x)
\end{align*}
and $V = W/(2\sigma^2)$.
\end{enumerate}
\end{lemma}
\begin{proof}
The rates of convergence can be calculated similarly to Proposition 6.4.4 in \cite{vdGeer1988}.
The limiting distribution follows along the same argument as the one in Example 6.6 of \cite{vdGeer1988}.
\end{proof}

\subsection{Asymptotic results for the class of smoothed PWAs}
\label{subsection_smoothed}
We aim to compute the LSE in (\ref{eq:P}).
However,
as $g_\theta$ is a non-smooth function (of $\theta$) we consider a smooth surrogate of $g_\theta$.
For $\mu >0$ let us define the following smooth approximation of $g_\theta$:
$$g_\theta^\mu(x) = \mu \log \left(e^{\frac{\alpha^\top x + \gamma}{\mu}} +
e^{\frac{\beta^\top x+ \phi}{\mu}}\right) - \mu \log 2
\quad \mbox{for } \theta = (\alpha,\beta, \gamma, \phi) \in \Theta \subset \R^{2d+2}.$$
Here are a few important facts about $g_\theta$ and $g_\theta^\mu$ \citep{Nesterov2005}:

\begin{itemize}
  \item[(i)] $\sup_{x \in \R^d} \left| g_\theta(x) - g_\theta^\mu(x) \right| \le \mu \log 2$, see \eqref{eq:EntProxApp}.
  \item[(ii)] The function $g_\theta^\mu$ is continuously differentiable with a gradient that is Lipschitz,
  given by $\mu^{-1} \max_{\Vert x \Vert_{d+1}=1} \langle [\alpha,\gamma],x \rangle^2 +  \langle [\beta,\phi],x \rangle^2$
  (see Lemma~\ref{lemma_nesterov}).
\end{itemize}

Suppose that $\{\mu_n\}_{n \ge 1}$ is a sequence of positive numbers such that $\mu_n = o(1)$.
Similarly to \eqref{eq:P},
we estimate $\theta_0$ using the method of least squares: 
\begin{equation}
	\hat \theta_n := \arg \min_{\theta \in \Theta} \frac{1}{n} \sum_{i=1}^n [Y_i - g_\theta^{\mu_n}(X_i)]^2.
\end{equation}
For the asymptotic results presented in the following theorem to hold true, amongst others,
it is necessary that the amount of smoothing
decreases at a $o \left( n^{-1/2} \right)$ rate, summarized in the next two assumptions:

\begin{assumption}[Parameter space]
\label{as:2}
Let $\Theta \subset \R^{2 d + 2}$ be a compact set such that $\theta_0 = (\alpha_0,\beta_0, \gamma_0, \phi_0)$
belongs to the interior of $\Theta$.
We assume $(\alpha_0,\beta_0) \ne (\gamma_0, \phi_0)$.
\end{assumption}

\begin{assumption}[Order of $\mu_n$]
\label{as:3}
Suppose $\{\mu_n\}_{n \ge 1}$ is a sequence of constants such that $\mu_n = o(n^{-1/2})$.
\end{assumption}

The main result is then summarized in the next theorem.
\begin{theorem}
\label{theorem_entropy_rates}
Under Assumptions~\ref{as:1}--\ref{as:3}, the following holds true.
\begin{enumerate}
  \item The least squares estimator is consistent and satisfies $\hat{\theta}_n-\theta_0= O_P(n^{-1/2})$.
  \item We have $$n^{1/2}(\hat \theta_n - \theta_0) \stackrel{d}{\to} N_d(0, V^{-1} W V^{-1}),$$ where 
$V := 2 P[\dot{g}_{\theta_0} \dot{g}_{\theta_0}^\top]$
and $W = P(\dot{m}_{\theta_0} \dot{m}_{\theta_0}^\top)  = 4 P[\eps^2 \dot{g}_{\theta_0} \dot{g}_{\theta_0}^\top].$
In particular, when $\eps$ is independent of $X$ with variance $\sigma^2 >0$, then $W = 2 \sigma^2 V$.
\end{enumerate}
\end{theorem}

\section{Experimental results}
\label{section_experimentals}
This section evaluates the proposed approach (Section~\ref{section_smoothed}) from a numerical perspective.
We start by exemplarily showing how our proposed approach performs in parameter estimation of a broken stick and a broken plane model (Section~\ref{subsection_example_1d2d}).
Further numerical evaluations address the dependence on the smoothing parameter $\mu$ (Section~\ref{subsection_dependence_smoothing_parameter}),
the ability of Algorithm~\ref{algorithm_smoothing} to find successfully minimize (\ref{eq:P}),
as well as an extension to estimation of PWAs with $k=3$ components (Section~\ref{subsection_extension_k3}).
In the entire section we used the implementation of the Newton method
provided by the \texttt{optim} function in $R$, obtained by setting its parameter \textit{method} to ``BFGS''.
Initial values for BFGS (the parameter $r$ in Algorithm~\ref{algorithm_smoothing} were always generated in the interval $[-1,1]$.

\subsection{Examples of parameter estimation in a broken stick and broken plane model}
\label{subsection_example_1d2d}
\begin{figure}
\centering
\begin{minipage}{0.49\textwidth}
\includegraphics[width=\textwidth]{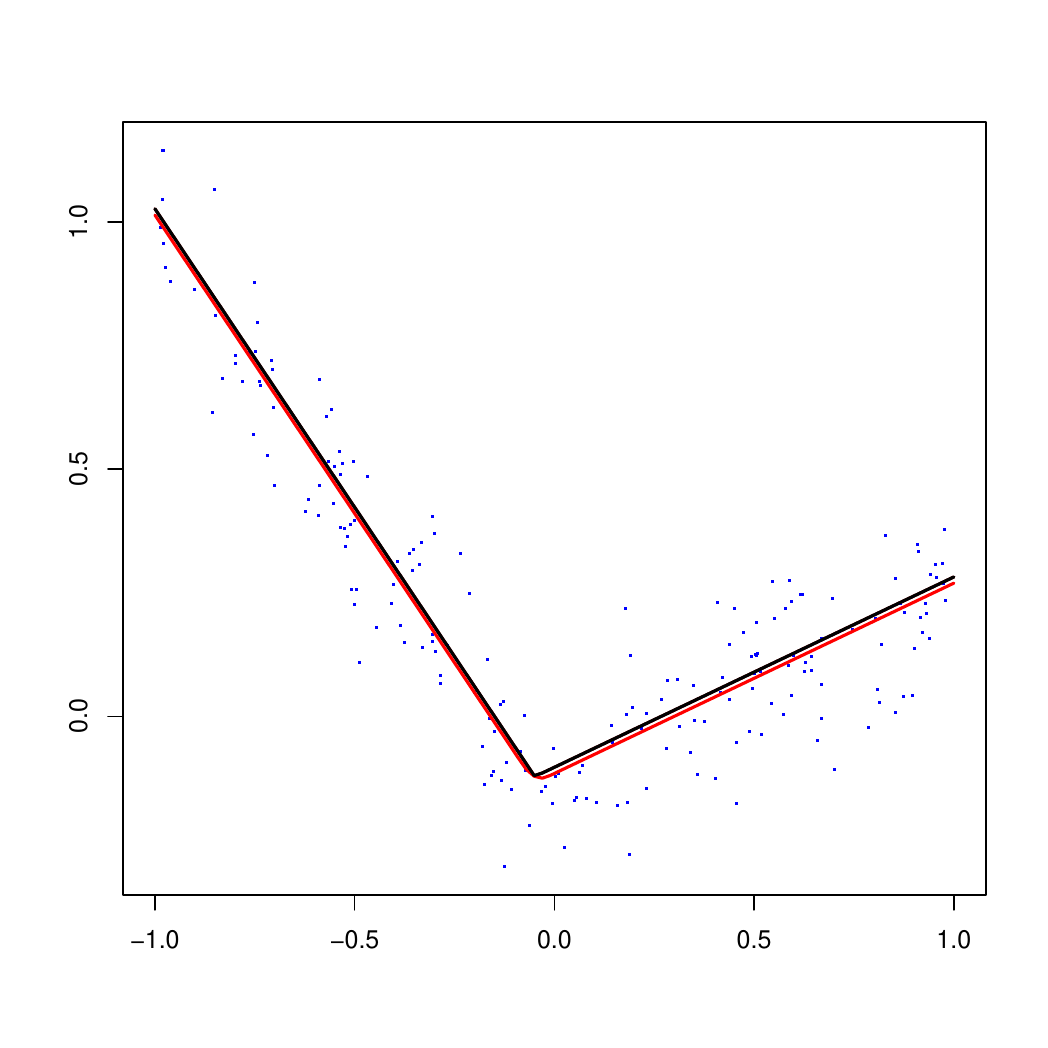}
\end{minipage}~
\begin{minipage}{0.49\textwidth}
\includegraphics[width=\textwidth]{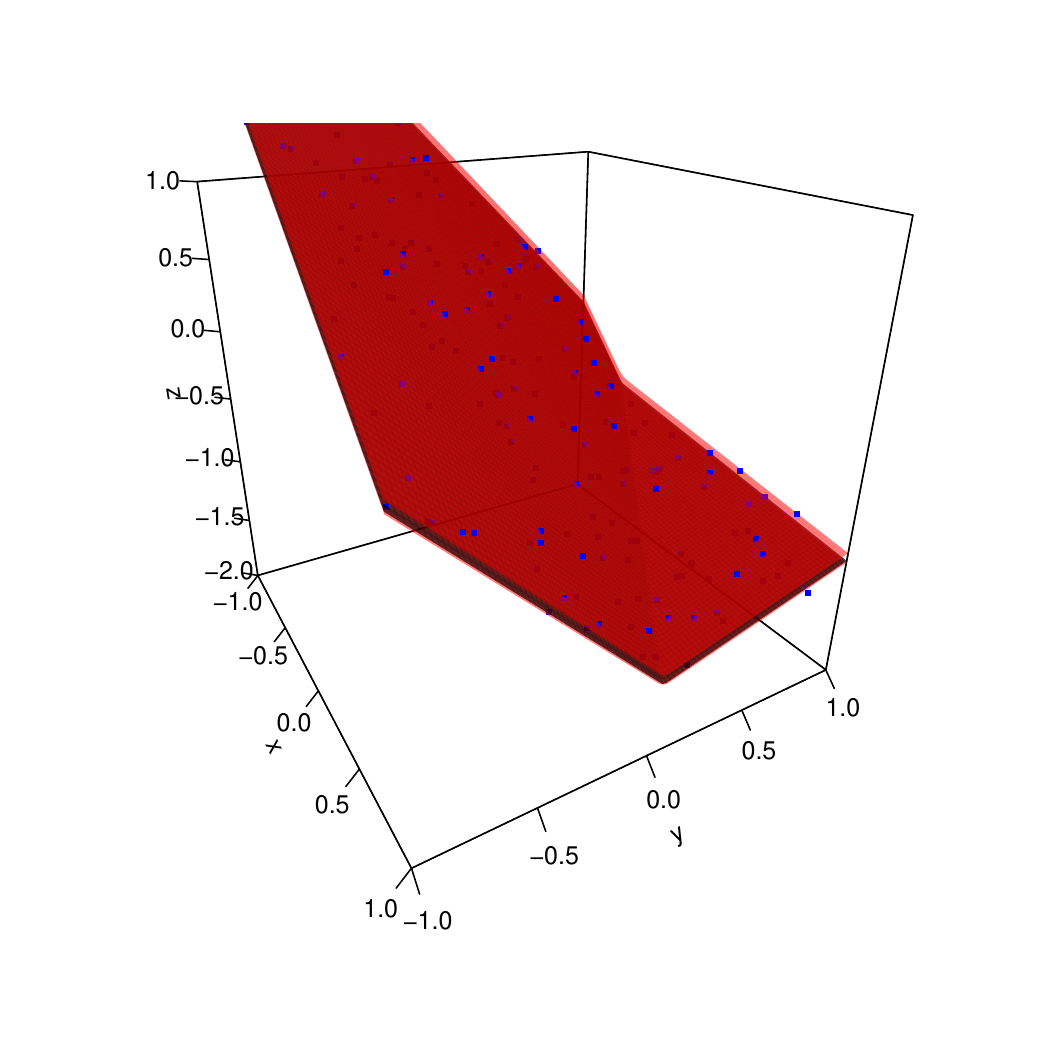}
\end{minipage}
\caption{Random points (blue) generated from a broken line (left) and broken plane (right) model with two components.
True model (darkred), \cite{Nesterov2005} smooth approximation (red) with $\mu=0.1$ and estimated model (black) from the data.}
\label{fig:example1d2d}
\end{figure}
We exemplarily show how Algorithm~\ref{algorithm_smoothing}
can be used in practice to estimate regression parameters.
We generated two parameter vectors, one for a broken stick and one for a broken plane model with two components, respectively.
We then simulated $n=200$ points from each model by uniformly selecting points on the lines (planes)
and by adding iid.\ normal noise with mean $0$ and standard deviation $0.1$.

We used Algorithm~\ref{algorithm_smoothing} as given in Section~\ref{section_algorithm}
in connection with squared error loss smoothing and a smoothing parameter $\mu=0.1$ to compute fitted parameters.

Figure~\ref{fig:example1d2d} shows the true (generated) model (darkred),
its smooth \cite{Nesterov2005} approximation (red) as well as the randomly generated points (blue).
The best fit obtained with Algorithm~\ref{algorithm_smoothing} (black)
shows that the original line (plane) can be reasonably well recovered in both cases.
We will elaborate on the quality and computability of such approximations in the following sections.

\subsection{Comparison of Algorithm~\ref{algorithm_smoothing} to a gradient free method}
\label{subsection_comparison}
\begin{figure}
\centering
\begin{minipage}{0.49\textwidth}
\includegraphics[width=\textwidth]{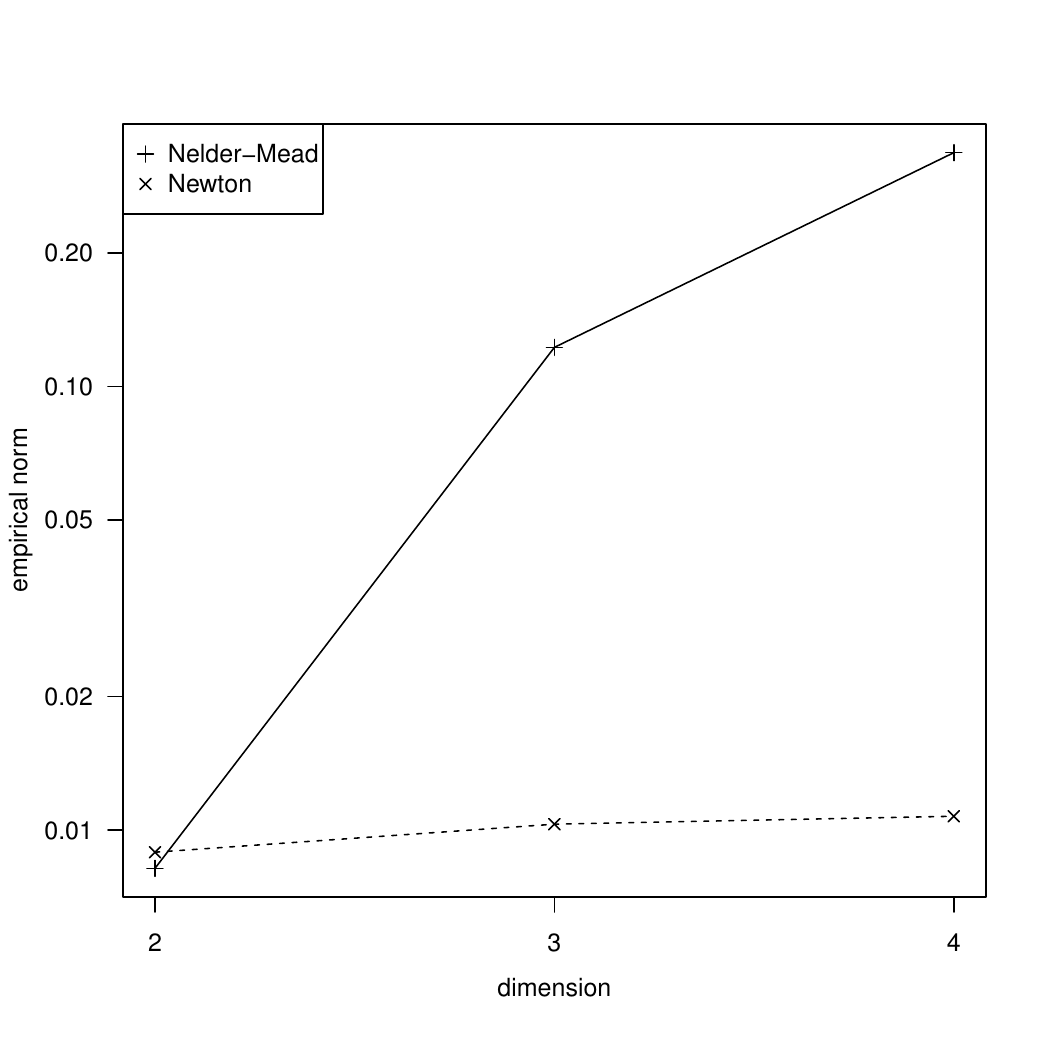}
\end{minipage}~
\begin{minipage}{0.49\textwidth}
\includegraphics[width=\textwidth]{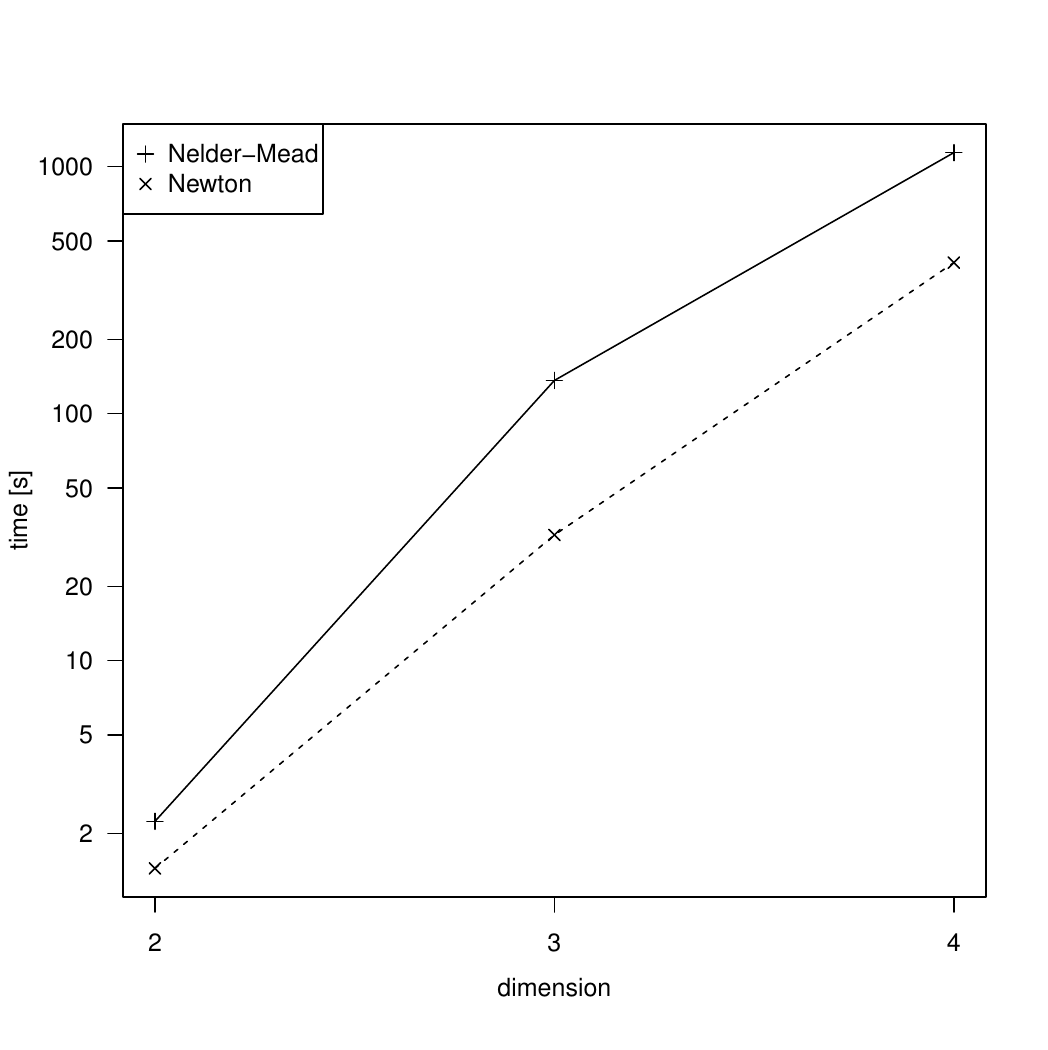}
\end{minipage}
\caption{Comparison of Algorithm~\ref{algorithm_smoothing} with the \cite{NelderMead1965} method
in terms of average empirical norm (left) and computational time (right) for the best fit.
Log scale on the time axis.}
\label{fig:plot_normtime}
\end{figure}
We compare the performance of Algorithm~\ref{algorithm_smoothing} to the popular gradient free method of \cite{NelderMead1965}.
For this we generate two planes in dimensions $d \in \{2,\ldots,4\}$ and make sure that these planes
intersect in the $[-1,1]^d$ cube with at least a $90$ degree angle in order to avoid singular cases.
We then generate $n=10^d$ points from each model in dimension $d$ and apply the \cite{NelderMead1965} method
as well as Algorithm~\ref{algorithm_smoothing} with squared error prox function and fixed smoothing parameter $\mu=0.1$.

We evaluate the quality of the fit using the empirical norm criterion as in Example~\ref{example_1d}.

Figure~\ref{fig:plot_normtime} shows results.
We see that in comparison to a gradient free method, using \cite{Nesterov2005} smoothing to obtain
gradients in Algorithm~\ref{algorithm_smoothing} yields a fraction of the computational runtime as the dimension increases.
The empirical norm increases for both methods as the dimension increases,
although the one of Algorithm~\ref{algorithm_smoothing} seems to consistently stay around one order of magnitude below
the one of the \cite{NelderMead1965} method
(note that the time axis has a log scale).

\subsection{Dependence on the smoothing parameter}
\label{subsection_dependence_smoothing_parameter}
\begin{figure}
\centering
\includegraphics[width=0.5\textwidth]{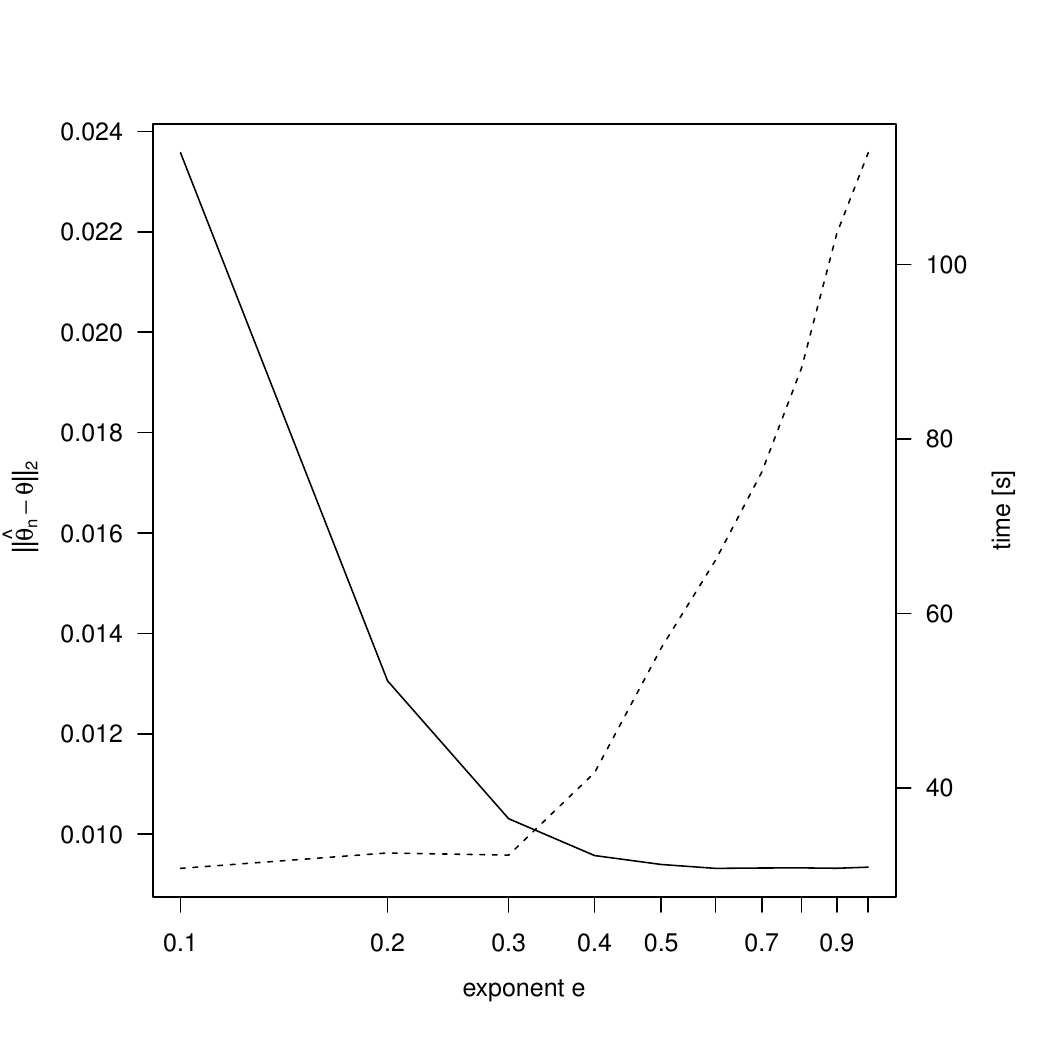}
\caption{Progression of $\left| \hat{\theta}_n-\theta_0\right|$ (solid line) as a function of $\mu=n^{-e}$ and corresponding time (dashed line) required for its computation. The number of points $n=1000$ was fixed and the exponent $e$ was varied.}
\label{fig:plot_mu}
\end{figure}
We further evaluate the dependence of the smoothing parameter on the performance of Algorithm~\ref{algorithm_smoothing}. For this we fix a parameter vector $\theta$ of a plane in two dimensions, generate $n=1000$ points on it and repeatedly fit a two plane segmented model to the data points using Algorithm~\ref{algorithm_smoothing} with squared error prox loss function. The smoothing parameter was chosen as $\mu=n^{-e}$, where the exponent $e \in \{0.1,0.2,\ldots,1\}$ ranged from $0.1$ to $1$ in steps of $0.1$.

For each $\mu$, we repeatedly fit the plane fixed at the start of the experiment a total number of $p=100$ times. From this pool we select the best estimate $\hat{\theta}_n$ measured in terms of the empirical norm and record its deviation $\left| \hat{\theta}_n-\theta \right|$ from $\theta$ as well as the computational time needed to compute $\hat{\theta}_n$.
This procedure is repeated a total number of $100$ times in order to average both the deviation as well as the computational time.

Figure~\ref{fig:plot_mu} shows averages of both $\left| \hat{\theta}_n-\theta \right|$ (solid line) as well as of the computational time (dashed line). As expected, the quality of the recovered plane increases (ie.\ the deviation of the estimate $\hat{\theta}_n$ from the true parameters $\theta$ goes to zero) as the parameter $\mu$ goes to zero. Likewise, the computational runtime to obtain estimates of this quality increases.

\subsection{Assessment of Newton's method and local minimization}
\label{subsection_newton_performance}
\begin{figure}
\centering
\begin{minipage}{0.49\textwidth}
\includegraphics[width=\textwidth]{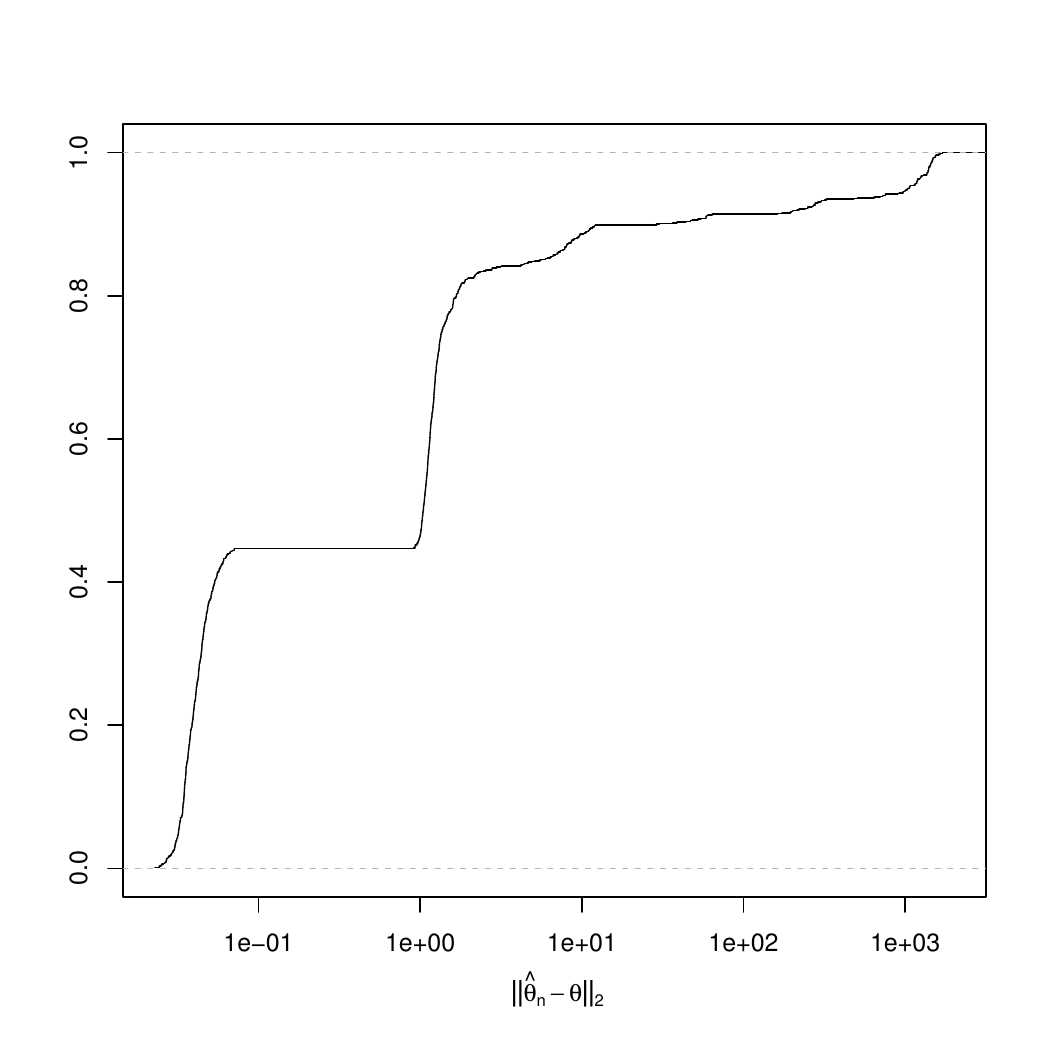}
\end{minipage}~
\begin{minipage}{0.49\textwidth}
\includegraphics[width=\textwidth]{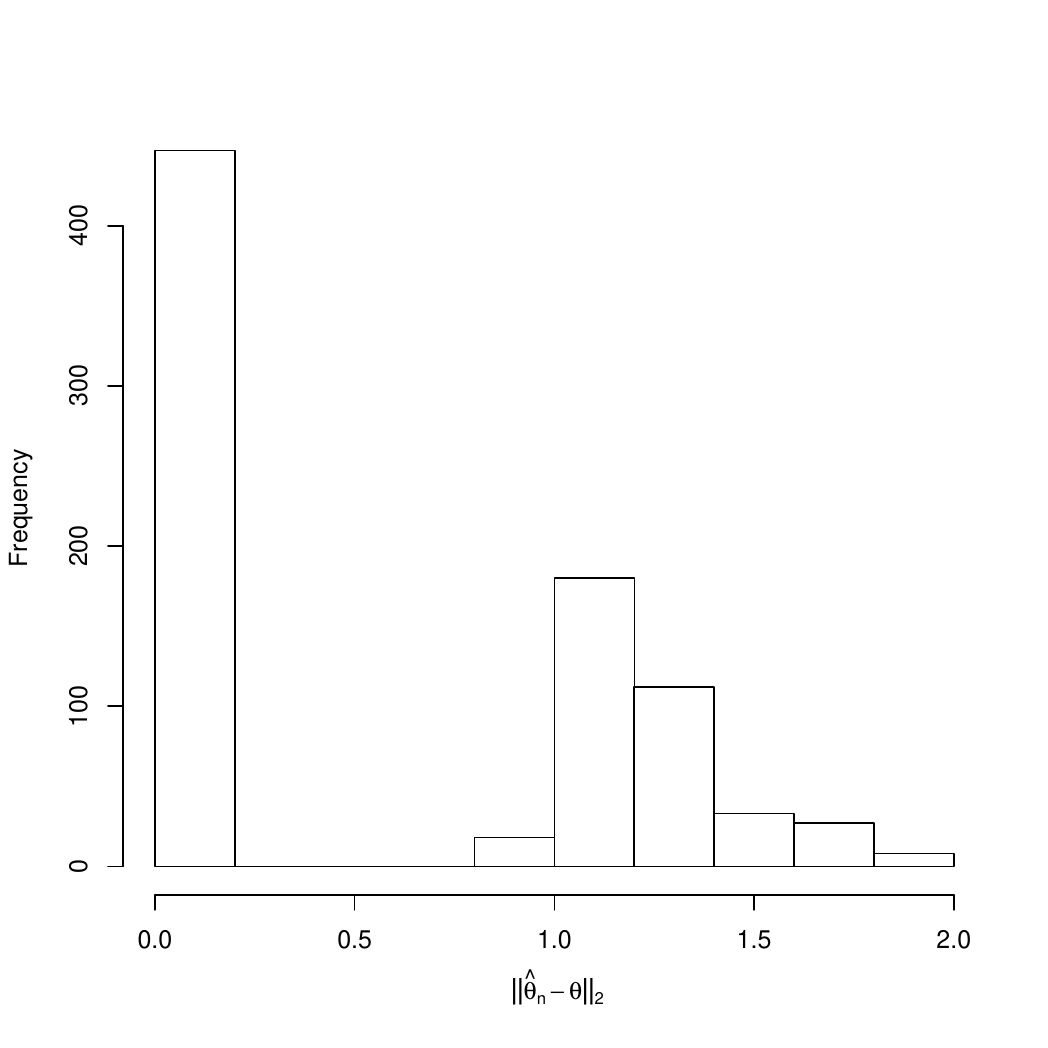}
\end{minipage}
\caption{Broken line regression with Algorithm~\ref{algorithm_smoothing} with squared error prox function and $\mu=0.1$. Results based on $1000$ repetitions.
Left: Ecdf of $\left| \hat{\theta}_n - \theta_0 \right|$, the deviation of the estimated from the generated parameters in the $l_2$ norm. Right: Histogram of $\left| \hat{\theta}_n - \theta_0 \right|$ after truncation at $0.1$.}
\label{fig:hist}
\end{figure}

\begin{figure}
\centering
\includegraphics[width=0.5\textwidth]{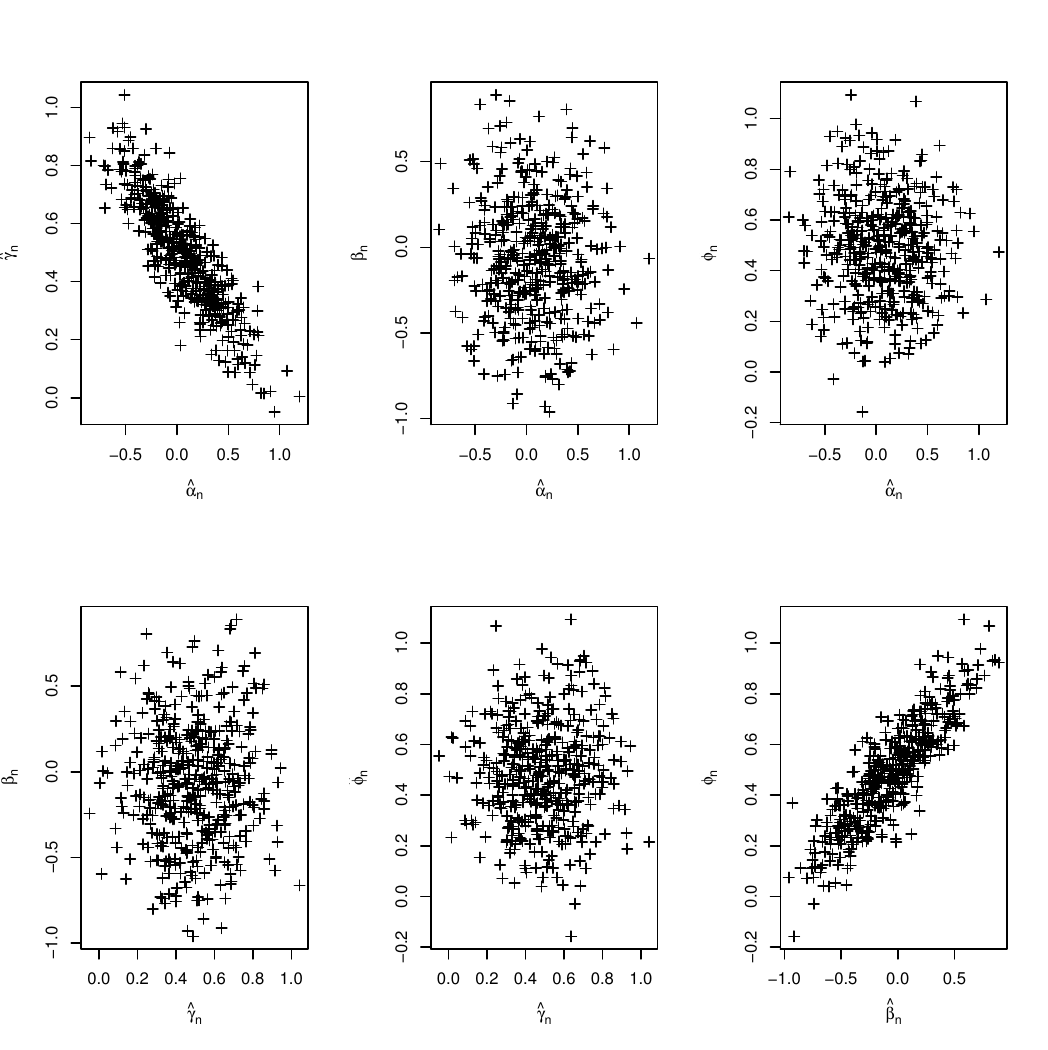}
\caption{Projections of all six pairs of parameter estimates of $\theta=(\alpha,\gamma,\beta,\phi) \in \R^4$.
Plot only shows estimates lying in the truncated range $[-1,1]$.}
\label{fig:multiplot}
\end{figure}

Figure~\ref{fig:hist} investigates how likely the Newton method in Algorithm~\ref{algorithm_smoothing} is able to find the global minimum in the case of a broken line regression (dimension one) with two components.

To this end we fix a broken plane model in one dimension and repeatedly apply Algorithm~\ref{algorithm_smoothing} using \cite{Nesterov2005} smoothing in connection with the squared error loss prox function and smoothing parameter $\mu=0.1$. All results are based on $1000$ repetitions.

Figure~\ref{fig:hist} (left) shows the ecdf of $\left| \hat{\theta}_n - \theta_0 \right|$ based on all $1000$ runs, that is the deviation of the estimated from the generated parameters in the $l_2$ norm. As can be seen from the figure, the Newton method converges to an estimated set of parameters very close to the true parameters in $l_2$ norm roughly in $90\%$ of all cases.
A separate histogram for all estimates $\hat{\theta}_n$ satisfying $\left| \hat{\theta}_nsubsection_extension_k3 - \theta_0 \right| < 0.1$ is displayed in Figure~\ref{fig:hist} (right).

Even though the Newton method converged to a very good approximation in most of the cases, the remaining estimates deviate from the (true) underlying parameters $\theta$ in $l_2$ norm by more than $10^3$.
An investigation of these cases showed that due to the particular implementation of the Newton method in the $R$ function \texttt{optim}, failure of converge results in the Newton run being aborted by some numerical criterion leading to value at the boundary of the parameter domain (of absolute magnitude up $10^3$ to $10^4$).

Figure~\ref{fig:multiplot} shows all six projections of two parameter estimates of $\theta=(\alpha,\gamma,\beta,\phi) \in \R^4$ in the $xy$-plane
for the $1000$ runs already analyzed in Figure~\ref{fig:hist}.
To capture the majority of estimates actually taking values of around zero, Figure~\ref{fig:multiplot} only shows estimates lying in the truncated range of $[-1,1]$.
It can be seen from the plot that a well defined optimum exists.

\subsection{Assessment of coverage probabilites}
\begin{table}[t]
\label{table_coverage}
\centering
\begin{tabular}{l||cc|cc}
Parameter & \multicolumn{2}{c}{\cite{SiegmundZhang1994}} & \multicolumn{2}{|c}{Algorithm~\ref{algorithm_smoothing}}\\
& coverage prob. & length & coverage prob. & length\\
\hline
$\theta_0$	&0.879 &0.169 &0.913 &0.184\\
$a_1$	&0.954 &0.208 &0.963 &0.225\\
$b_1$	&0.959 &0.119 &0.972 &0.128\\
$a_2$	&0.957 &0.219 &0.976 &0.242\\
$b_2$	&0.951 &0.128 &0.969 &0.141
\end{tabular}
\caption{Coverage probabilites and lengths of the confidence intervals for $\theta_0=(a_1,b_1,a_2,b_2)$ and its four components.
Model consisting of two lines in one dimension (one parametrized in $(a_1,b_1)$ and the other in $(a_2,b_2)$).
Exact method of \cite{SiegmundZhang1994} and Algorithm~\ref{algorithm_smoothing}.}
\end{table}
We assess the accuracy of confidence intervals computed for all parameters,
measured in terms of their coverage probabilites and lengths.
To this end, we compare confidence intervals computed with Algorithm~\ref{algorithm_smoothing}
to the ones computed with the exact method of \cite{SiegmundZhang1994}.

Our setup is as follows: We fix a PWA consisting of two lines in one dimension, characterized through a $\theta_0$.
We then repeat the following procedure $R=1000$ times:
\begin{enumerate}
  \item Using Algorithm~\ref{algorithm_smoothing} (with $\mu=0.01$) and the exact method by \cite{SiegmundZhang1994},
we compute the best estimate of $\hat{\theta}_n$, measured with respect to the empirical norm, from a pool of $10$ repetitions.
For each of the $10$ repetitions, we generate $200$ points on the PWA.
For the method of \cite{SiegmundZhang1994}, we determine the change point in $[-1,1]$ of the two lines
via a grid search with $1000$ equidistant points in $[-1,1]$.

  \item For the best estimate determined in the previous step,
we store $\hat{\theta}_n$, the set $P$ of $200$ points generated on the PWA
as well as the empirical covariance matrix $C$ computed with $\hat{\theta}_n$ and $P$.
The matrix $C$ is computed using plug-in estimates for the integrals in matrix $W$ of Lemma~\ref{lemma_unsmoothed},
where the points in $P$ are divided up into two subsets corresponding
to the two sides of the estimated change-point specified in $\hat{\theta}_n$.

  \item Finally, a normal confidence interval is computed for
the $i$'th component of $\theta_0$ as $\hat{\theta}_n^{(i)} \pm 1.96 \sqrt{C_{i,i}/N_i}$ for all $i \in \{1,\ldots,4\}$,
where $\hat{\theta}_n^{(i)}$ is the $i$'th entry of the estimate $\hat{\theta}_n$ and
$N_i$ is the number of points in $P$ which fall into the line segment defined through $\hat{\theta}_n^{(i)}$.
\end{enumerate}
In each repetition of the above procedure we record
(1) the number of times each of the four confidence intervals contains the truth in $\theta_0$
as well as the number of times all confidence intervals simultaneously contains all entries of $\theta_0$
(this allows us to compute coverage probabilites)
and (2) the length of all confidence intervals.

Results are shown in Table~\ref{table_coverage}.
The table demonstrates that for the individual components of $\theta_0$,
the exact method of \cite{SiegmundZhang1994} yields confidence intervals which keep the $95\%$ coverage
while being only slightly conservative.
Algorithm~\ref{algorithm_smoothing} yields individual confidence intervals
which are slightly too wide, hence resulting in higher lengths than the exact ones and coverage probabilites which exceed $95\%$.
With respect to the simultaneous coverage for the entire vector $\theta_0$,
the method of \cite{SiegmundZhang1994} yields a confidence interval with a lower coverage probability than the one of
Algorithm~\ref{algorithm_smoothing}.

\subsection{Extension to three and more planes}
\label{subsection_extension_k3}
\begin{figure}
\centering
\includegraphics[width=0.5\textwidth]{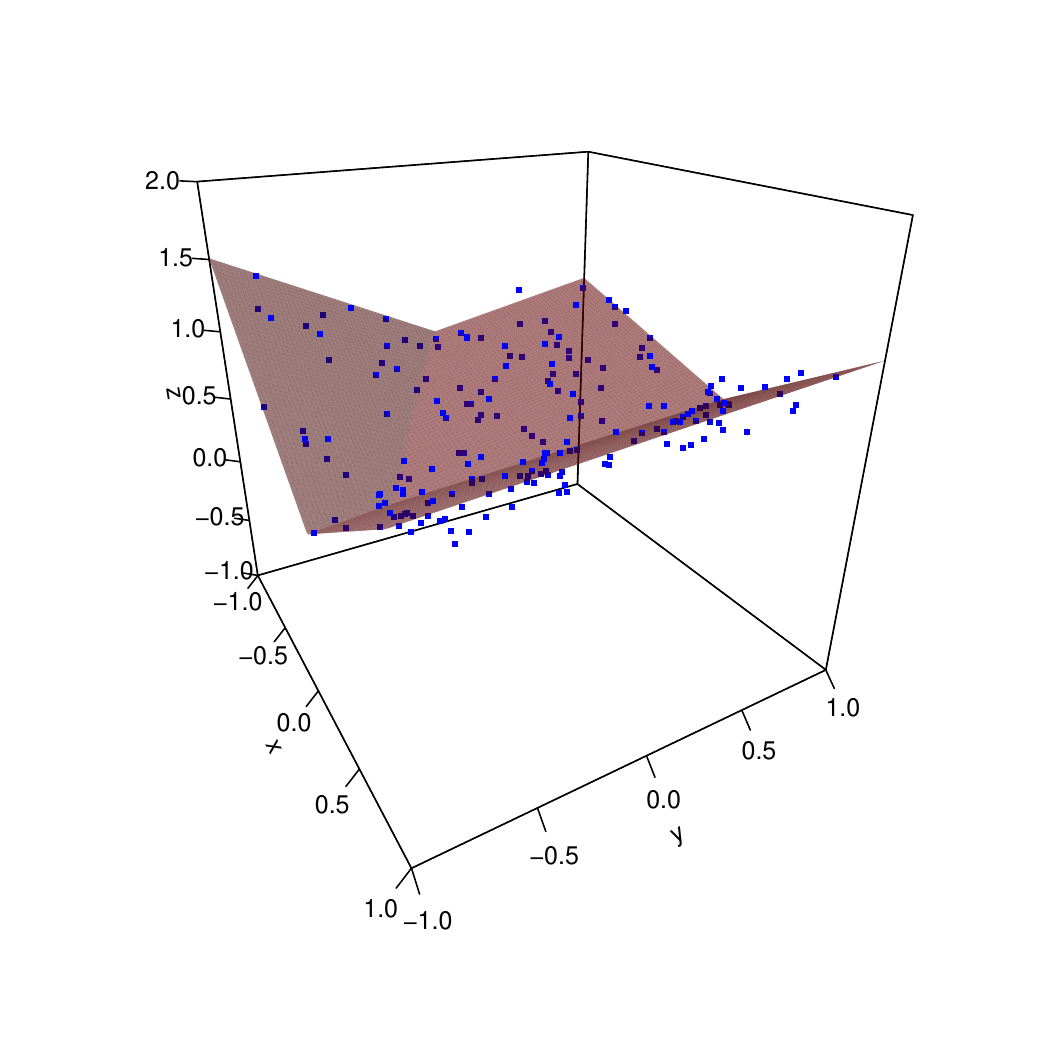}
\caption{Fit of parameters for a segmented plane consisting of three components.
Algorithm~\ref{algorithm_smoothing} with squared error prox smoothed function.
Blue points are generated points from the true model (red), fitted segmented planes in black.}
\label{fig:plot3planes}
\end{figure}
An extension to three planes can be seen in Figure~\ref{fig:plot3planes}.

\newpage
\appendix
\section{Proofs of Section~\ref{subsection_smoothed}}
\label{section_proofs_smoothed_entropy}
This section proves that the least squares estimate of the entropy smoothed functions
$g_\theta^\mu$ are $\sqrt{n}$ consistent and asymptotically normal.
This is the result of Theorem~\ref{theorem_entropy_rates}.
It will be proven by means of the following lemmas and theorems.

\begin{lemma}\label{lem:Consistency}
Under Assumptions~\ref{as:1}-\ref{as:3} we have $\hat \theta_n -\theta_0 = o_p(1)$.
\end{lemma}
\begin{proof}
For $\theta \in \Theta$, let us define $$M(\theta) := P[(Y - g_\theta(X))^2]\qquad \mbox{and} \qquad M_n(\theta) = P[(Y - g_\theta^{\mu_n}(X))^2].$$ Also, let $$\M_n(\theta) := \P_n[(Y - g_\theta^{\mu_n}(X))^2] = \frac{1}{n} \sum_{i=1}^n [Y_i - g_\theta^{\mu_n}(X_i)]^2.$$ We will apply Theorem 3.2.3 of VdV\&W to prove this result. In particular, we will show that $\sup_{\theta \in \Theta}| \M_n(\theta) - M(\theta)| \to 0$ in probability. Observe that $$ \M_n(\theta) - M(\theta) =  [\M_n(\theta) - M_n(\theta)] +  [M_n(\theta) - M(\theta)].$$ Now, \begin{eqnarray}
|(M_n - M)(\theta)| & = & \left|P[(Y - g_\theta^{\mu_n}(X))^2] - P[(Y - g_\theta(X))^2] \right| \nonumber \\
& = & \left| P\left[\{g_\theta^{\mu_n}(X) - g_\theta(X)\}\cdot \{-2 Y + g_\theta^{\mu_n}(X) + g_\theta(X) \} \right] \right| \nonumber \\
& = & \left| P\left[\{g_\theta^{\mu_n}(X) - g_\theta(X)\}\cdot \{-2 g_{\theta_0}(X) + g_\theta^{\mu_n}(X) + g_\theta(X) \} \right] \right| \nonumber \\
& \le & P\left[|g_\theta^{\mu_n}(X) - g_\theta(X)|\cdot |-2 g_{\theta_0}(X) + g_\theta^{\mu_n}(X) + g_\theta(X) | \right] \nonumber \\
& \le & (\mu_n \log 2) \; P[|g_\theta^{\mu_n}(X) - g_{\theta}(X)| + 2 |g_\theta(X) - g_{\theta_0}(X)|]  \nonumber \\
& \le & (\mu_n \log 2) \; P[ (\mu_n \log 2) + 2 G(X) \|\theta - \theta_0\|] \label{eq:M_n-M-1}
\end{eqnarray}
where we have used the facts:
$$\sup_{x \in \R^d} |g_\theta(x) - g_\theta^{\mu_n}(x)| \le \mu_n \log 2,\qquad \mbox{for all } \theta \in \Theta,$$ and 
\begin{equation}\label{eq:Lip_G}
|g_\theta(x) - g_{\theta_0}(x)| \le G(x) \|\theta - \theta_0\|,
\end{equation} for $G(x) = \|(1, x)\|$ (see Lemma~\ref{lemma_gLipschitz}).
Thus, by Assumptions~\ref{as:1} and~\ref{as:3}, $$\sup_{\theta \in \Theta}| M_n(\theta) - M(\theta)| \stackrel{a.s.}{\to} 0 \qquad \mbox{ as } \; n \to \infty.$$

Now, for any $\eta >0$, $$\P \left(\sup_{\theta \in \Theta}| \M_n(\theta) - M_n(\theta)| > \eta \right) \le \frac{1}{\eta} \E \left[\sup_{\theta \in \Theta}| \M_n(\theta) - M_n(\theta)| \right] \le J_{[\;]}(1) \|F\|_{P,2},$$ 
where $J_{[\;]}(1) := $

\end{proof}
The following result gives the rate of convergence of the estimator $\hat \theta_n$.
\begin{theorem}
\label{tmr:ConvRate}
Under Assumptions~\ref{as:1}-\ref{as:3} we have $\hat \theta_n -\theta_0 = O_p(n^{-1/2})$.
\end{theorem}
\begin{proof}
We will apply Theorem 3.2.5 of VdV and W to prove the result. 

Note that $\theta_0$ minimizes $M(\theta)$ for $\theta \in \Theta$. Hence, as $M(\theta)$ is strictly convex around $\theta_0$, we have $$M(\theta)- M(\theta_0) \ge c \|\theta - \theta_0\|^2, \qquad \mbox{for $\theta$ in a neighborhood of $\theta_0$},$$ for some $c >0$.

We now have to bound 
\begin{equation}\label{eq:ExpSup}
\E \left[\sup_{\|\theta - \theta_0\| < \delta} \sqrt{n} |(\M_n - M)(\theta) - (\M_n - M)(\theta_0)| \right].\end{equation} 
Observe that
\begin{eqnarray*}
	\sqrt{n} |(\M_n - M)(\theta) - (\M_n - M)(\theta_0)| \le S_n(\theta) + T_n(\theta) 
\end{eqnarray*}
where 
\begin{eqnarray*}
S_n(\theta) & := & \sqrt{n} |(\M_n - M_n)(\theta) - (\M_n - M_n)(\theta_0)| \\
T_n(\theta) & := & \sqrt{n} |(M_n - M)(\theta) - (M_n - M)(\theta_0)|.
\end{eqnarray*}
The stochastic process $S_n(\theta)$ can be bounded using metric entropy calculations. Next we have to control  the deterministic process $T_n(\theta)$. Similarly, we can show that
\begin{eqnarray}
|(M_n - M)(\theta_0)| & = & \left|P[(Y - g_{\theta_0}^{\mu_n}(X))^2] - P[(Y - g_{\theta_0}(X))^2] \right| \nonumber\\
& = & P\left[|g_{\theta_0}^{\mu_n}(X) - g_{\theta_0}(X)|^2\right] \\
& \le & \mu_n^2 (\log 2)^2. \label{eq:M_n-M-2}
\end{eqnarray}
Therefore, the expected supremum in~\eqref{eq:ExpSup} can be upper bounded by 
\begin{eqnarray*}
& & \E \left[\sup_{\|\theta - \theta_0\| < \delta} S_n(\theta) \right] + \E \left[\sup_{\|\theta - \theta_0\| < \delta} T_n(\theta) \right] \\
& \le & c_0 \delta + \sqrt{n} [c_1 \mu_n^2 + c_2 \mu_n \delta] =: \phi_n(\delta),
\end{eqnarray*}
where $c_0$ can be found using metric entropy, $c_1 = 2 (\log 2)^2$ and $c_2 = 2 (\log 2) P[G(X)]$. Note that $\phi_n(\delta)/\delta^\alpha$ is a decreasing function for any $1 < \alpha < 2$. Then with $r_n := n^{1/2}$, we have $$r_n^2 \phi_n(r_n^{-1}) = c_0 r_n + c_1 \sqrt{n}  r_n^2 \mu_n^2 + c_2 r_n \mu_n \lesssim \sqrt{n}, \qquad \mbox{ for every } n,$$ as $r_n \mu_n = O(1)$.

As $\hat \theta_n \to_p \theta_0$ (by Lemma~\ref{lem:Consistency}) then by Theorem 3.2.5 of VdV\&W we have $$n^{1/2}(\hat \theta_n - \theta_0) = O_p(1).$$
\end{proof}

Let us define $$m_\theta(x,y) := (y - g_\theta(x))^2.$$ Observe that $(x,y) \mapsto m_\theta(x,y)$ is a measurable function for each $\theta \in \Theta$ and that $\theta \mapsto m_\theta(x,y)$ is differentiable at $\theta_0$ for $P$-a.e.~$x$ with derivative $$\dot{m}_{\theta_0}(x,y) = -2 (y - g_{\theta_0}(x)) \dot{g}_{\theta_0}(x).$$ 
\begin{theorem}
\label{tmr:limit}
Under assumptions~\ref{as:1}-\ref{as:3} we have $$n^{1/2}(\hat \theta_n - \theta_0) \stackrel{d}{\to} N_d(0, V^{-1} W V^{-1}),$$ where 
\begin{equation}\label{eq:Def_V}
V := 2 P[\dot{g}_{\theta_0} \dot{g}_{\theta_0}^\top]
\end{equation}
 and $$W = P(\dot{m}_{\theta_0} \dot{m}_{\theta_0}^\top)  = 4 P[\eps^2 \dot{g}_{\theta_0} \dot{g}_{\theta_0}^\top].$$
In particular, when $\eps$ is independent of $X$ with variance $\sigma^2 >0$, then $W = 2 \sigma^2 V$.
\end{theorem}
\begin{proof}
For the convenience of notation, let us write $$m_\theta^{\mu}(x,y) := (y - g_\theta^{\mu}(x))^2,\;\; (\mbox{for } \mu >0).$$
We will study the stochastic process
$$\tilde \M_n(h) := n \P_n(m_{\theta_0 + h n^{-1/2}}^{\mu_n} - m_{\theta_0}^{\mu_n}), \qquad \mbox{for }\; h \in \R^{2d +2}.$$ Observe that $$\sqrt{n} (\hat \theta_n - \theta_0) = \argmin_{h}  \tilde \M_n(h).$$ 

We will show that $$\tilde \M_n(h) \stackrel{d}{\to} h^\top \Delta + \frac{1}{2} h^\top V h =: \tilde \M(h) \;\; \mbox{in } \ell^\infty(\{h: \|h\|\le K\})$$ for every $K>0$, where $\Delta \sim N_d(0, W)$. Then the conclusion of the theorem follows from the argmax (argmin) continuous theorem (see e.g., Theorem 3.2.2 in VdV\&W) upon noticing that $$\sqrt{n} (\hat \theta_n - \theta_0) = \argmin_{h}  \tilde \M_n(h) \stackrel{d}{\to} \argmin_h \tilde \M(h) = - V^{-1} \Delta \sim N_d(0, V^{-1} W V^{-1}).$$

First observe that 
\begin{eqnarray}
n \P_n(m_{\theta_0 + h n^{-1/2}}^{\mu_n} - m_{\theta_0}^{\mu_n}) & = & n \P_n(m_{\theta_0 + h n^{-1/2}} - m_{\theta_0}) \nonumber \\ 
&& \; + \; n\P_n[(m_{\theta_0 + h n^{-1/2}}^{\mu_n} - m_{\theta_0}^{\mu_n}) - (m_{\theta_0 + h n^{-1/2}} - m_{\theta_0})]. \label{eq:Split}
\end{eqnarray}
{\bf Study of the second term on the right-hand side of~\eqref{eq:Split}:} The second term on the right-hand side of the above display can also be expressed as $$n (\M_n - M) (\theta_0 + h n^{-1/2}) - n (\M_n - M) (\theta_0).$$
Using a similar expansion as in~\eqref{eq:M_n-M-1} with $\theta = \theta_0 + h n^{-1/2}$ (and $\P_n$ instead of $P$) we have $$|n(\M_n - M) (\theta_0 + h n^{-1/2})| \le n \mu_n^2 (\log 2)^2 + 2 n \mu_n (\log 2) \P_n[G(X)] \|n^{-1/2} h\|$$ which converges uniformly to 0 a.s.~in $\ell^\infty(\{h: \|h\|\le K\})$ as $n^{-1/2} \mu_n \to 0$ and $\P_n[G(X)] \le P[G(X)] + 1 < \infty$ for all large $n$ a.s. (by Assumption~\ref{as:1}). Now using a similar calculation as in~\eqref{eq:M_n-M-2}, we have $$|n (\M_n - M) (\theta_0)| \le n \mu_n^2 (\log 2)^2 \to 0, \qquad \mbox{as }\; n \to \infty.$$ Thus, we have shown that the second term in~\eqref{eq:Split} goes to zero (a.s.) uniformly in $T_K := \{h \in \R^{2 d + 2}: \|h\|\le K\}$.

{\bf Study of the first term on the right-hand side of~\eqref{eq:Split}:} Observe that
\begin{eqnarray}
n \P_n(m_{\theta_0 + h n^{-1/2}} - m_{\theta_0}) & =  & \sqrt{n}(\P_n - P) [\sqrt{n} (m_{\theta_0 + h n^{-1/2}} - m_{\theta_0})] \nonumber \\ 
& & \qquad +  \; n P(m_{\theta_0 + h n^{-1/2}} - m_{\theta_0}).\label{eq:Split-2}
\end{eqnarray} 
By a second order Taylor expansion of $M(\theta) := P [m_\theta]$ about $\theta_0$, we have 
$$P [m_\theta] - P[m_{\theta_0}] = P[(g_\theta - g_{\theta_0})^2] = \frac{1}{2} (\theta - \theta_0)^\top V (\theta - \theta_0) + o(\|\theta -\theta_0\|^2),$$ where $V$ is defined in~\eqref{eq:Def_V}. Thus, the second term of the right side of the~\eqref{eq:Split-2} converges to $(1/2) h^\top V h$ uniformly on $T_K$. 

To handle the first term in~\eqref{eq:Split-2} we use the Donsker's theorem with classes of functions changing with $n$; see e.g., Theorem 2.11.23 of VdV\&W. For $K >0$ fixed, we consider the following class of functions $$\F_n := \{\sqrt{n}(m_{\theta_0 + h n^{-1/2}}- m_{\theta_0}): \|h \| \le K\}.$$ Notice that for all $\theta_1, \theta_2$ in a neighborhood $\mathcal{N}$ of $\theta_0$, $$|m_{\theta_1}(x,y) - m_{\theta_2}(x,y)| = |g_{\theta_1}(x) - g_{\theta_2}(x)|\cdot |-2 y + g_{\theta_1}(x) + g_{\theta_2}(x)|  \le F(x,y) \|\theta_1 - \theta_2\|$$ where $$F(x,y) := G(x) \left[2 |y| + \sup_{\theta_1,\theta_2 \in \mathcal{N}} |g_{\theta_1}(x) + g_{\theta_2}(x)| \right]$$ and $G$ is defined after~\eqref{eq:Lip_G}. Further note that $F \in L_2(P)$ (by Assumption~\ref{as:1}). 

Thus, the function class $\F_n$ has envelope $F_n(x,y) \equiv K F(x,y)$ for all $n$, and since $F(x,y) \in L_2(P)$ the Lindeberg conditions (see Equation~2.11.21 of VdV\&W) are satisfied easily. Furthermore, for $s,t \in \R^{2d +2}$ such that $\|s\|,\|t\| < K$, define $$f_{n,s} :=  \sqrt{n}(m_{\theta_0 + s n^{-1/2}} - m_{\theta_0}), \quad \mbox{and}\quad f_{n,t} := \sqrt{n}(m_{\theta_0 + t n^{-1/2}} - m_{\theta_0}).$$ By the dominated convergence theorem the covariance functions satisfy $$P(f_{n,g} f_{n,t}) - P(f_{n,s}) P(f_{n,t}) \to P(s^\top \dot{m}_{\theta_0} \dot{m}_{\theta_0}^\top t) = s^\top W t.$$ Lastly, to apply Theorem 2.11.23 of VdV\&W we need to that the bracketing entropy condition holds. To this end, observe that $$N_{[\;]}(2 \eta \|F_n\|_{P,2}, \mathcal{F}_n,L_2(P)) \le N(\eta, T_K, \|\cdot \|) \le \left(\frac{CK}{\eta}\right)^d,$$ where the first inequality follows from Theorem 2.7.11 of VdV\&W (classes that are Lipschitz in a parameter) and the second inequality follows from an upper bound on the $\eta$-covering number of an Euclidean ball in $\R^{2d+ 2}$. Thus, $$\int_0^\delta \sqrt{N_{[\;]}(2 \eta \|F_n\|_{P,2}, \mathcal{F}_n,L_2(P))} d \eta \lesssim \int_0^\delta \sqrt{d \log \left( \frac{CK }{\eta}\right)} d\eta \to 0 \qquad \mbox{as } \delta \to 0,$$ and hence the bracketing entropy condition holds. We conclude that $\tilde \M_n(h)$ converges weakly to $h^\top \Delta$ in $\ell^\infty(\{h: \|h\|\le K\})$, and the desired result holds.
\end{proof}

The proof of Theorem~\ref{theorem_entropy_rates} is merely a summary of the aforementioned results:
\begin{proof}[Proof of Theorem~\ref{theorem_entropy_rates}]
The first statement follows from Lemma~\ref{lem:Consistency} and Theorem~\ref{tmr:ConvRate}.
The second statement follows from Theorem~\ref{tmr:limit}.
\end{proof}

\subsection{Additional lemmas}
\begin{lemma}
\label{lemma_gLipschitz}
It holds true that $|g_\theta(x)-g_{\theta_0}(x)| \leq G(x) \Vert \theta-\theta_0 \Vert$.
\end{lemma}
\begin{proof}
Consider $|g_\theta(x) - g_{\theta_0}(x)|$ for a fixed $x$ and four cases depending on whether
$g_\theta(x) = \max\{\alpha^\top x + \gamma, \beta^\top x+ \phi\}$
takes its maximum in the first or second argument.
\begin{enumerate}
  \item If $g_\theta(x) = \alpha^\top x + \gamma$, $g_{\theta_0}(x) = \alpha_0^\top x + \gamma_0$,
  then immediately $|g_\theta(x) - g_{\theta_0}(x)| \leq (\Vert x \Vert+1) \Vert \theta-\theta_0 \Vert$
  by triangle and Cauchy-Schwarz inequalities.
  The norm is the Euclidean norm.
  \item The case $g_\theta(x) = \alpha^\top x + \gamma$, $g_{\theta_0}(x) = \beta_0^\top x + \phi_0$.
  If $\alpha^\top x + \gamma \leq \beta_0^\top x + \phi_0$ then
  $|g_\theta(x) - g_{\theta_0}(x)| \leq |(\beta^\top x+ \phi) - (\beta_0^\top x + \phi_0)| \leq (\Vert x \Vert+1) \Vert \theta-\theta_0 \Vert$
  due to the fact that $\beta^\top x+ \phi \leq \alpha^\top x + \gamma$.
  Likewise, if $\alpha^\top x + \gamma > \beta_0^\top x + \phi_0$ then using
  $\beta_0^\top x + \phi_0 \geq \alpha_0^\top x + \gamma_0$
  one obtains
  $|g_\theta(x) - g_{\theta_0}(x)| \leq |(\alpha^\top x+ \gamma) - (\alpha_0^\top x + \gamma_0)| \leq (\Vert x \Vert+1) \Vert \theta-\theta_0 \Vert$.
  \item The case $g_\theta(x) = \beta^\top x + \phi$, $g_{\theta_0}(x) = \alpha_0^\top x + \gamma_0$ follows analogously to the previous one.
  \item The case $g_\theta(x) = \beta^\top x + \phi$, $g_{\theta_0}(x) = \beta_0^\top x + \phi_0$ follows analogously to the first one.
\end{enumerate}
Combined this shows that
$|g_\theta(x) - g_{\theta_0}(x)| \le G(x) \|\theta - \theta_0\|$ with $G(x)=\Vert x \Vert+1$.
\end{proof}

According to \cite{vdVaartWellner2000},
the quantity
$$S_n(\theta) := \sqrt{n} |(\M_n - M_n)(\theta) - (\M_n - M_n)(\theta_0)|$$
is bounded by
$$E^\ast \int_0^{\theta_n(\delta)} \sqrt{\log N(\epsilon,M_{n,\delta},L_2(\P_n)} d\epsilon.$$

\begin{lemma}
\label{lemma_entropy_entropyG}
Assuming $\E(X^2) < \infty$,
$H(\delta,\G_n(R),Q_n) \sim \log \left( \frac{C + \delta}{\delta} \right)$
for a constant $C$ independent of $\mu$.
\end{lemma}
\begin{proof}
Define
$$\G' = \left\{ g_\theta(x) = \frac{\alpha^T x+\gamma}{\mu}: \R^d \rightarrow \R ~|~
(\alpha,\beta,\gamma,\phi) \in \Theta \right\}.$$

Let $\{ c_j \}_{j=1}^N$ be a covering of $\Theta$ with balls $B_j := B(c_j,\epsilon)$ centered at $c_j$
with radius $\epsilon$,
and for each $c_j=(\alpha_j,\beta_j,\gamma_j,\phi_j)$, define the corresponding function
in $\G'$ as $C_j(x):=g_{c_j}(x)$.

Let $\delta>0$ be given.
Then according to \cite[Lemma 2.5]{vdGeer2009},
$\Theta \subseteq \R^{2d+2}$ (since it is bounded by a ball of radius $\tilde{R}$) can be covered by
$$N \leq \left( \frac{4\tilde{R}+\epsilon}{\epsilon} \right)^{2d+2}$$
balls of radius $\epsilon$.

For $g_0 \in \G'$, let $g \in \G'_n(R) = \{ g \in \G': \Vert g-g_0 \Vert_{Q_n} \leq R \}$,
where $\Vert g_0 \Vert_{Q_n} \leq \mu^{-1}\tilde{R}(M+1)$
using the bounds on $\alpha,\beta,\gamma,\phi$ and $M$.
It follows that $\Vert g \Vert_{Q_n} \leq \Vert g-g_0 \Vert_{Q_n} + \Vert g_0 \Vert_{Q_n} \leq R+\mu^{-1}\tilde{R}(M+1) =: R_2$. 
As $g \in \G'$, $g(x) = \mu^{-1}(\alpha^T x+\gamma)$ for some $(\alpha,\beta,\gamma,\phi) \in \Theta$.

For the $(\alpha,\beta,\gamma,\phi) \in \Theta$ defining $g$, there exist $(\alpha_j,\beta_j,\gamma_j,\phi_j)$
such that
$\Vert (\alpha,\beta,\gamma,\phi) - (\alpha_j,\beta_j,\gamma_j,\phi_j) \Vert \leq \epsilon$,
which in turn implies that $\Vert \alpha-\alpha_j \Vert \leq \epsilon$,
$\Vert \beta-\beta_j \Vert \leq \epsilon$, $|\gamma-\gamma_j| \leq \epsilon$, $|\phi-\phi_j| \leq \epsilon$.

It follows immediately that
$|(g-C_j)(x)| \leq \frac{1}{\mu} \epsilon (M+1)$, implying
$$\Vert g-C_j \Vert_{Q_n}^2 = \frac{1}{n} \sum_{i} |(g-C_j)(x_i)|^2 \leq \frac{1}{\mu^2} \epsilon^2 (M+1)^2 = \delta^2$$
for the choice $\epsilon=\frac{\mu\delta}{M+1}$.
Hence, $\{ C_j \}_{j=1}^N$ form a $\delta$-cover of $\G'_n(R_2)$ and
$$H(\delta,\G'_n(R),Q_n) \leq (2d+2) \log \left( \frac{4 R_2(\mu) (M+1)+\mu\delta}{\mu\delta} \right),$$
which is of the form $\log \left( \frac{C + \delta}{\delta} \right)$
for a constant $C$.

Equally, the class
$$\G'' = \left\{ g(x) = \frac{\beta^T x+\phi}{\mu}: \R^d \rightarrow \R ~|~
(\alpha,\beta,\gamma,\phi) \in \Theta \right\}$$
has the same entropy.

Let $F_1=G'$, $F_2=G''$ and $\phi(x,y) = f(x,y) = \mu \log \left[ \frac{1}{2} ( e^x+e^y) \right]$.
Next, compute the entropy for the class $\phi(F_1,F_2)$ using \cite[Lemma 9.13]{Kosorok2008}.

Both classes $F_1$ and $F_2$ are BUEI,
and $\phi:\R^2 \rightarrow \R$ is locally Lipschitz on $\Theta$,
meaning that it satisfies the requirement
$|\phi \circ f(x) - \phi \circ g(x)|^2 \leq c^2 \sum_{j=1}^k (f_j(x)-g_j(x))^2 = c^2 \left| f(x)-g(x) \right|^2$
of the lemma.

The proof of \cite[Lemma 9.13]{Kosorok2008} then shows that the covering number of $\phi(F_1,F_2)$ can be bounded
by the product of the individual covering numbers, and hence that
the entropy of $\phi(F_1,F_2)$ satisfies
\begin{align*}
&\log N(\delta H,\phi(F_1,F_2),Q_n)\\
&\leq \log N(\delta \Vert F_1 \Vert_{Q,2},F_1,Q_n) + \log N(\delta \Vert F_2 \Vert_{Q,2},F_2,Q_n),
\end{align*}
for the combined envelope $H=|\phi(f_0)|+c(|f_{01}|+F_1+|f_{02}|+F_2)$.

Hence
$H(\delta,\G_n(R),Q_n) \sim \log \left( \frac{C + \delta}{\delta} \right) + \log \left( \frac{C + \delta}{\delta} \right)
\sim \log \left( \frac{C + \delta}{\delta} \right)$.
\end{proof}

\subsection{Computation of the covariance matrices for squared error prox}
Note that $\theta=(\alpha,\gamma,\beta,\phi)$ (in this order),
$m_\theta(x) = (y-g_\theta^\mu(x))^2$ and
$$\frac{\partial}{\partial \theta} m_\theta(x) =
(-2) (y-g_\theta^\mu(x)) \left( e^\frac{\alpha^T x+\gamma}{\mu} + e^\frac{\beta^T x+\phi}{\mu} \right)^{-1}
\begin{pmatrix} x e^\frac{\alpha^T x+\gamma}{\mu} \\ e^\frac{\alpha^T x+\gamma}{\mu}\\
x e^\frac{\beta^T x+\phi}{\mu} \\ e^\frac{\beta^T x+\phi}{\mu} \end{pmatrix}.$$
Let $A:=\left( e^\frac{\alpha^T x+\gamma}{\mu} + e^\frac{\beta^T x+\phi}{\mu} \right)$.
Hence, using $\E(y-g_\theta^\mu(x))=\epsilon$,
\begin{align*}
W &= P [\dot{m}_\theta \dot{m}_\theta^\top] = \int_{-1}^1 4 (y-g_\theta^\mu(x))^2 A^{-2}
\begin{pmatrix} x e^\frac{\alpha^T x+\gamma}{\mu} \\ e^\frac{\alpha^T x+\gamma}{\mu}\\
x e^\frac{\beta^T x+\phi}{\mu} \\ e^\frac{\beta^T x+\phi}{\mu} \end{pmatrix}
\begin{pmatrix} x e^\frac{\alpha^T x+\gamma}{\mu} \\ e^\frac{\alpha^T x+\gamma}{\mu}\\
x e^\frac{\beta^T x+\phi}{\mu} \\ e^\frac{\beta^T x+\phi}{\mu} \end{pmatrix}^\top dx\\
&= 4 \epsilon^2 \int_{-1}^1 A^{-2}
\begin{pmatrix} xx^\top e^{2\frac{\alpha^T x+\gamma}{\mu}} & x e^{2\frac{\alpha^T x+\gamma}{\mu}} &
xx^\top e^\frac{\alpha^T x+\gamma}{\mu} e^\frac{\beta^T x+\phi}{\mu} & x e^\frac{\alpha^T x+\gamma}{\mu} e^\frac{\beta^T x+\phi}{\mu}\\
&\ddots\\
&&\ddots\\
&&&\ddots
\end{pmatrix} dx
\end{align*}
assuming independence of $\epsilon$ from $X$.

Similarly, $V$ is the Jacobi matrix of $\E(y-g_\theta^\mu(x))^2$. Here,
$$\nabla P m_\theta = \E \left( (-2)(y-g_\theta^\mu(x)) A^{-1} \left[ x e^\frac{\alpha^T x+\gamma}{\mu}, e^\frac{\alpha^T x+\gamma}{\mu}, x e^\frac{\beta^T x+\phi}{\mu}, e^\frac{\beta^T x+\phi}{\mu} \right] \right).$$
By the product rule, the entry $(1,1)$ in $V = JP m_\theta$ (Jacobi matrix of $P m_\theta$) is
$$\int_{-1}^1 (+2) \left( A^{-1} x e^\frac{\alpha^T x+\gamma}{\mu} \right)^2 + (-2) \epsilon
\left[ \frac{x^2}{m} A^{-1} e^\frac{\alpha^T x+\gamma}{\mu} - \frac{x^2}{m} A^{-2} e^{2\frac{\alpha^T x+\gamma}{\mu}} \right] dx,$$
the $(1,2)$ entry is
$$\int_{-1}^1 (+2) x \left( A^{-1} e^\frac{\alpha^T x+\gamma}{\mu} \right)^2 + (-2) \epsilon
\left[ \frac{x}{m} A^{-1} e^\frac{\alpha^T x+\gamma}{\mu} - \frac{x}{m} A^{-2} e^{2\frac{\alpha^T x+\gamma}{\mu}} \right] dx,$$
the $(1,3)$ entry is
$$\int_{-1}^1 (+2) x A^{-1} e^\frac{\alpha^T x+\gamma}{\mu} x A^{-1} e^\frac{\beta^T x+\phi}{\mu} + (-2) \epsilon
\frac{-x^2}{m} A^{-2} e^\frac{\alpha^T x+\gamma}{\mu} e^\frac{\beta^T x+\phi}{\mu} dx,$$
and the $(1,4)$ entry is
$$\int_{-1}^1 (+2) x A^{-1} e^\frac{\alpha^T x+\gamma}{\mu} A^{-1} e^\frac{\beta^T x+\phi}{\mu} + (-2) \epsilon
\frac{-x}{m} A^{-2} e^\frac{\alpha^T x+\gamma}{\mu} e^\frac{\beta^T x+\phi}{\mu} dx.$$
Similarly for the other entries.

As $\mu \rightarrow 0$, in the region where $\max g_\theta(x) = \alpha^T x + \gamma$, that is where $\alpha^T x + \gamma >> \beta^T x + \phi$ (corresponding to the first two rows of the matrices $W$ and $V$), $A$ grows asymptotically like $e^\frac{\alpha^T x+\gamma}{\mu}$,
hence $A^{-1} e^\frac{\alpha^T x+\gamma}{\mu}$ goes to one and $A^{-2} e^\frac{\alpha^T x+\gamma}{\mu} e^\frac{\beta^T x+\phi}{\mu}$ goes to zero.

Comparing with $V$ and $W$ in Lemma~\ref{lemma_unsmoothed} shows that the matrices for the unsmoothed case are recovered as $\mu \rightarrow 0$.

\subsection{Computation of the covariance matrices for squared error prox}
Note that $\theta=(\alpha,\gamma,\beta,\phi)$ (in this order), $m_\theta(x) = (y-g_\theta^\mu(x))^2$ and
$$\frac{\partial}{\partial \theta} m_\theta(x) = (-2) (y-g_\theta^\mu(x)) \begin{pmatrix} w_1 x \\ w_1 \\ w_2 x \\ w_2 \end{pmatrix}.$$
Hence, using $\E(y-g_\theta^\mu(x))=\epsilon$,
\begin{align*}
W &= P[\dot{m}_\theta \dot{m}_\theta^\top] = \int_{-1}^1 4 (y-g_\theta^\mu(x))^2
\begin{pmatrix} w_1 x \\ w_1 \\ w_2 x \\ w_2 \end{pmatrix}
\begin{pmatrix} w_1 x \\ w_1 \\ w_2 x \\ w_2 \end{pmatrix}^\top dx\\
&= 4 \epsilon^2 \int_{-1}^1
\begin{pmatrix} w_1^2 xx^\top & w_1^2 x & w_1 w_2 xx^\top & w_1 w_2 x\\
&\ddots\\
&&\ddots\\
&&&\ddots
\end{pmatrix} dx
\end{align*}
assuming independence of $\epsilon$ from $X$.

Similarly, $V$ is the Jacobi matrix of $\E(y-g_\theta^\mu(x))^2$.
Here,
$$\nabla P m_\theta = \E \left( (-2)(y-g_\theta^\mu(x))
\left[ w_1 x, w_1, w_2 x, w_2 \right] \right),$$
and the Jacobi matrix is
\begin{align*}
V &= JP m_\theta = \int_{-1}^1 (+2)
\begin{pmatrix} w_1^2 xx^\top & w_1^2 x & w_1 w_2 xx^\top & w_1 w_2 x\\
w_1^2 x & w_1^2 & w_1 w_2 x & w_1 w_2\\
w_1 w_2 xx^\top & w_1 w_2 x & w_2^2 xx^\top & w_2^2 x\\
w_1 w_2 x & w_1 w_2 & w_2^2 x & w_2^2
\end{pmatrix}dx\\
&= 2 \begin{pmatrix} w_1^2 \int_{-1}^1 xx^\top dx & w_1^2 \int_{-1}^1 x dx & w_1 w_2 \int_{-1}^1 xx^\top dx & w_1 w_2 \int_{-1}^1 x dx\\
&\ddots\\
&&\ddots\\
&&&\ddots
\end{pmatrix}
\end{align*}

As $\mu \rightarrow 0$, in the region where $\max g_\theta(x) = \alpha^T x + \gamma$
(corresponding to the first two rows of the matrices $W$ and $V$)
$w_1 \rightarrow 1$ and $w_2 \rightarrow 0$,
and conversely for the case $\max g_\theta(x) = \beta^T x + \phi$.
Comparing with $V$ and $W$ in Lemma~\ref{lemma_unsmoothed}
shows that the matrices for the unsmoothed case are recovered as $\mu \rightarrow 0$.

\section{A generalization of the model in \cite{SiegmundZhang1994} for the intersection of two planes}
\label{section_siegmund_two}
\subsection{Two lines}
\label{subsection_siegmund_two_lines}
The first part derives the model of \cite{SiegmundZhang1994} from two line equations.
A priori, consider the model of two lines and a change point $\theta$,
$$y = (\alpha_1 x +\alpha_2) \I(x<\theta) + (\beta_1 x + \beta_2) \I(x \geq \theta).$$
Leaving out either indicator is valid since then, the remaining coefficients will be fit to compensate
for the first (fixed) line equation, ie.
$$y = (\alpha_1 x +\alpha_2) + (\beta_1 x + \beta_2) \I(x \geq \theta).$$
Since $\theta$ is the change point, both lines must intersect at $\theta$, hence
$$\alpha_1 \theta + \alpha_2 = (\alpha_1+\beta_1)\theta + (\alpha_2+\beta_2),$$
since the for $x \geq \theta$, the second indicator is one and the new terms are added to the first line.
Simplifying the above yields $\beta_2=-\beta_1 \theta$, and plugging this into the line equation
yields
\begin{align*}
y &= (\alpha_1 x +\alpha_2) + (\beta_1 x - \beta_1 \theta) \I(x \geq \theta)\\
&= (\alpha_1 x +\alpha_2) + \beta_1 (x - \theta) \I(x \geq \theta)\\
&= \alpha_1 x +\alpha_2 + \beta_1 (x - \theta)^+,
\end{align*}
where it was used that
\cite{SiegmundZhang1994} define $(x-\theta)^+ = \max(x-\theta,0)$, which is $(x-\theta)$ if $x \geq \theta$ and $0$ otherwise.
Thus it is the same as $(x - \theta) \I(x \geq \theta)$.

\subsection{Two planes}
\label{subsection_siegmund_two_planes}
For 3D, a similar expression for two intersecting planes can be obtained in a similar fashion.
First, a priori the model of two intersecting planes is given by
$$z = (\alpha_1 x +\alpha_2 y + \alpha_3) \I(P1) + (\beta_1 x + \beta_2 y + \beta_3) \I(P2),$$
where the indicators symbolically encode whether $(x,y)$ lie in plane $P1$ or $P2$.

Note first that similar to the 2D case, any two intersecting planes in 3D can be separated by a vertical
plane that is inserted at the right spot in-between them.
It is not necessary to tilt such a separator plane, similarly to the fact that a one-dimensional parameter $\theta$
was sufficient to separate two lines.
Precisely, the separator plane has to be inserted between the two intersecting planes vertically such that
it stands on the projection line of the intersection line of the two planes (projected down to the lower square side
of the cube in which the two planes are defined).

Let $C$ be the lower square side of the cube in which the two planes are defined.
Let $P$ be the rim of this square.
By considering any two points $p,q \in \partial P$,
all separator planes can be parameterized, and any point $(x,y)$ can be characterized to lie
on either the ``left'' or ``right'' side of the separator by considering the cross product.

Define for any point $(x,y)$ the side indicator
$$(x,y)^s = \text{sign} \left( \left| \begin{matrix} q_x-p_x & x-p_x\\ q_y-p_y & y-p_y \end{matrix} \right| \right).$$
This quantity is $1$ if $(x,y)$ is one the ``left'' side of the line through $p$ and $q$,
it is $-1$ if $(x,y)$ is on the ``right'' side and it is $0$ if $(x,y)$ lies on the line.
Using the side indicator,
$$z = (\alpha_1 x +\alpha_2 y + \alpha_3) \I((x,y)^s<0) + (\beta_1 x + \beta_2 y + \beta_3) \I((x,y)^s\geq 0).$$

Define analogously to \cite{SiegmundZhang1994},
$$\phi(x,y)^+ = \phi(x,y) \I((x,y)^s \geq 0) = \begin{cases} \phi(x,y) & (x,y)^s \geq 0,\\ 0 & \text{otherwise}, \end{cases}$$
i.e.\ the result of $\phi(x,y)^+$ is precisely $\phi(x,y)$ if its argument $(x,y)$ is on the $+$ (left) side.

Using the fact that the parameters of the second plane can compensate for the first ones, the first indicator
can again be omitted, leaving
$$z = (\alpha_1 x +\alpha_2 y + \alpha_3) + (\beta_1 x + \beta_2 y + \beta_3) \I((x,y)^s\geq 0).$$

Now, suppose both planes intersect above the projected line through $p,q$, that is all
points in the intersection of the planes additionally satisfy the line equation
$$(x-p_x)(q_y-p_y) = (y-p_y)(q_x-p_x) \Rightarrow y = \frac{(x-p_x)(q_y-p_y)}{q_x-p_x}+p_y =: f(x).$$
This assumes that $p_x \neq q_x$. Otherwise, the line through $p,q$ can be parameterized as
$$(x-p_x)(q_y-p_y) = (y-p_y)(q_x-p_x) \Rightarrow x = \frac{(y-p_y)(q_x-p_x)}{q_y-p_y}+p_x = p_x =: f(y).$$

As in the 2D case, use this to save one parameter of the plane equation: On the line $(x,y=f(x))$,
the two planes given above by $z=...$ with one omitted indicator coincide and hence
$$\alpha_1 x + \alpha_2 f(x) + \alpha_3 = (\alpha_1 + \beta_1)x + (\alpha_2+\beta_2)f(x) + (\alpha_3+\beta_3),$$
leading to $\beta_3=-\beta_1 x - \beta_2 f(x)$.
Simplifying the plane equation gives
\begin{align*}
z &= (\alpha_1 x +\alpha_2 y + \alpha_3) + (\beta_1 x + \beta_2 y -\beta_1 x - \beta_2 f(x)) \I((x,y)^s\geq 0)\\
&= \alpha_1 x +\alpha_2 y + \alpha_3 + \beta_2 (y - f(x)) \I((x,y)^s\geq 0)\\
&= \alpha_1 x +\alpha_2 y + \alpha_3 + \beta_2 (y - f(x))^+,
\end{align*}
for $p_x \neq q_x$, and using the definition of $f(x,y)^+$.

This is the analogue of the \cite{SiegmundZhang1994} model for intersecting planes. It makes sense to have $4$ parameters since
for a given projected intersection line, one needs three parameters to define the first plane,
the precise non-projected intersection follows from a vertical cut through the first plane at the given projected line,
and only one parameter is then needed to define the angle with which the second plane arises from the non-projected intersection line.

Do the two planes in the model coincide at $(x,f(x))$? Clearly, the first plane above the given projected intersection line
is given by $\alpha_1 x + \alpha_2 f(x) + \alpha_3$, whereas the other is given by
$\alpha_1 x + \alpha_2 f(x) + \alpha_3 + \beta_2 (f(x)-f(x))^+ = \alpha_1 x + \alpha_2 f(x) + \alpha_3$.
They hence intersect at $(x,f(x))$.

Reversely, if the two planes in the above model intersect then
$$\alpha_1 x + \alpha_2 y + \alpha_3 = \alpha_1 x + \alpha_2 y + \alpha_3 + \beta_2 (y-f(x))$$
and hence at the intersection $y=f(x)$, provided $\beta_2 \neq 0$, hence the two planes intersect
precisely at the pre-specified projected intersection line $(x,f(x)))$.

\bibliographystyle{apalike}
\bibliography{PhaseRegMultiDim}
\end{document}